\documentclass[11pt]{amsart}
\usepackage{enumerate}
\usepackage{hyperref}
\usepackage{amsmath,amssymb,amsthm}
\usepackage{tikz}
\usepackage{algorithm}%
\usepackage{algpseudocode}
\usepackage[margin=1.3in]{geometry}
\usepackage{todonotes}

\hypersetup{
    pdftitle = {Polynomial kernel for distance-hereditary deletion},
   pdfauthor=  {Eun Jung Kim, and O-joung Kwon}
   }

\newcommand\abs[1]{\lvert #1\rvert}

\newtheorem{THM}{Theorem}[section]
\newtheorem{LEM}[THM]{Lemma}

\newtheorem{PROP}[THM]{Proposition}

\newtheorem{OBS}[THM]{Observation}
\newtheorem{CLAIM}{Claim}

\theoremstyle{remark}

\newenvironment{proofofclaim}{\noindent \textsc{Proof of the Claim:}}{\hfill$\Diamond$\medskip}
\newtheorem{RRULE}{Reduction Rule}

\newcommand{\YES}{\textsc{Yes}}
\newcommand{\NO}{\textsc{No}}

\theoremstyle{definition}

\newcommand\dist{\operatorname{dist}}

\newcommand\dhd{\textsc{Distance-hereditary Vertex Deletion}}
\newcommand\sdhd{\textsc{DH Vertex Deletion}}
\newcommand\dhm{{DH-modulator}}

\begin{document}
\title[A polynomial kernel for distance-hereditary vertex deletion]{A polynomial kernel for \\ distance-hereditary vertex deletion}
\author{Eun Jung Kim}
\author{O-joung Kwon }
\address[Kim]{CNRS-Universit\'{e} Paris-Dauphine, Place du Marechal de Lattre de Tassigny, 75775 Paris cedex 16, France}
\address[Kwon]{Institute of Software Technology and Theoretical Computer Science, Technische Universit\"at Berlin, Germany.}
\email{eunjungkim78@gmail.com}
\email{ojoungkwon@gmail.com}
\thanks{The second author was supported by ERC Starting Grant PARAMTIGHT (No. 280152) and also supported by the European Research Council (ERC) under the European Union's Horizon 2020 research and innovation programme (ERC consolidator grant DISTRUCT, agreement No. 648527).}
\date{\today}
\begin{abstract}

A graph is \emph{distance-hereditary}  if for any pair of vertices, their distance in every connected induced subgraph containing both vertices is the same as their distance in the original graph. 
The \textsc{Distance-Hereditary Vertex Deletion} problem asks, given a graph $G$ on $n$ vertices and an integer $k$, whether there is a set $S$ of at most $k$ vertices in $G$
such that $G-S$ is distance-hereditary. 
This problem is important due to its connection to the graph parameter rank-width that distance-hereditary graphs are exactly graphs of rank-width at most $1$.
 Eiben, Ganian, and Kwon (MFCS' 16) proved that \textsc{Distance-Hereditary Vertex Deletion} can be solved in time $2^{\mathcal{O}(k)}n^{\mathcal{O}(1)}$, and asked whether 
 it admits a polynomial kernelization.
We show that this problem admits a polynomial kernel, answering this question positively.
 For this, we use a similar idea for obtaining an approximate solution for \textsc{Chordal Vertex Deletion} due to Jansen and Pilipczuk (SODA' 17) to obtain an approximate solution with $\mathcal{O}(k^3\log n)$ vertices when the problem is a \YES-instance, and
we exploit the structure of split decompositions of distance-hereditary graphs to reduce the total size.
\end{abstract}
\maketitle

\section{Introduction}

The graph modification problems,  in which we want to transform a graph to satisfy a certain property with as few graph modifications as possible, have been extensively studied. For instance, the \textsc{Vertex Cover} and \textsc{Feedback Vertex Set} problems are graph modification problems where the target graphs are edgeless graphs and forests, respectively. By the classic result of Lewis and Yannakakis~\cite{Lewis1980}, it is known that for all non-trivial hereditary properties that can be tested in polynomial time, the corresponding vertex deletion problems are NP-hard. Hence, the research effort has been directed toward designing  algorithms such as approximation and parameterized algorithms. 

	When the target graph class $\mathcal{C}$ admits efficient algorithms for some NP-hard problems, 
	the graph modification problem related to such a class attracts more attention. 
	In this context, vertex deletion problems  to classes of graphs of constant tree-width or constant tree-depth have been studied. 
	\textsc{Tree-width $w$ Vertex Deletion}
	\footnote{\textsc{Tree-width (or Tree-depth) $w$ Vertex Deletion} asks, given a graph $G$ and an integer $k$, whether 
	$G$ contains a vertex set $S$ of at most $k$ vertices such that $G-S$ has tree-width (or tree-depth) at most $w$.} 
	is proved to admit an FPT algorithm running in time $2^{\mathcal{O}(k)}n^{\mathcal{O}(1)}$ and a kernel 
	with $\mathcal{O}(k^{g(w)})$ vertices for some function $g$~\cite{FominLMS2012, KimLPRRSS2016}.
	Also, it was shown that \textsc{Tree-depth $w$ Vertex Deletion} admits uniformly polynomial kernels with $\mathcal{O}(k^6)$ vertices, for every fixed $w$~\cite{GiannopoulouJLS2015}.
	All these problems are categorized as vertex deletion problems for $\mathcal{F}$-minor free graphs in a general setting, when the set $\mathcal{F}$ contains at least one planar graph. 
	However, $\mathcal{F}$-minor free graphs capture only sparse graphs in a sense that the number of edges of such a graph is bounded by a linear function on the number of its vertices. Thus these problems are not very useful when dealing with very dense graphs.
		
	\emph{Rank-width}~\cite{Oum05} and \emph{clique-width}~\cite{CourcelleO2000} are graph width parameters introduced for extending graph classes of bounded tree-width.
	Graphs of bounded rank-width represent graphs that can be recursively decomposed along vertex partitions $(X,Y)$ 
	where the number of neighborhood types between $X$ and $Y$ are small. 
	Thus, graphs of constant rank-width may contain dense graphs; for instance, all complete graphs have rank-width at most $1$. 
	Courcelle, Makowski, and Rotics~\cite{CourcelleMR2000}
	proved that every $MSO_1$-expressible problem can be solved in polynomial time on graphs of bounded rank-width.

	Motivated from \textsc{Tree-width $w$ Vertex Deletion}, Eiben, Ganian, and the second author~\cite{EibenGK2016} initiated study on 
	vertex deletion problems to graphs of constant rank-width.
	The class of graphs of rank-width at most $1$ is exactly same as the class of distance-hereditary graphs~\cite{Oum05}.
	A graph $G$ is called \emph{distance-hereditary} if for every connected induced subgraph $H$ of $G$ and every two vertices $u$ and $v$ in $H$, 
	the distance between $u$ and $v$ in $H$ is the same as the distance in $G$. 
	A vertex subset $X$ of a graph $G$ is a \emph{distance-hereditary modulator}, or a \emph{DH-modulator} in short, 
	if $G-X$ is a distance-hereditary graph. We formulate our central problem.

\vskip 0.2cm
\noindent
\fbox{\parbox{0.97\textwidth}{
\dhd\ (\sdhd)\\
\textbf{Input :} A graph $G$, an integer $k$ \\
\textbf{Parameter :} $k$ \\
\textbf{Question :} Does $G$ contain a \dhm\ of size at most $k$? }}
\vskip 0.2cm

	Eiben, Ganian, and the second author~\cite{EibenGK2016} proved that \sdhd\ can be solved in time $2^{\mathcal{O}(k)}n^{\mathcal{O}(1)}$.
	It was known before that vertex deletion problems for graphs of rank-width $w$ can be solved in FPT time~\cite{KanteKKP15} using the fact that 
	graphs of rank-width at most $w$ can be characterized by a finite list of forbidden vertex-minors~\cite{Oum05}, and
	Eiben et al. devised a first elementary algorithm for this problem when $w=1$.
	Furthermore, they discussed that the size $k$ of a DH-modulator
	can be used to obtain a $2^{\mathcal{O}(k)}n^{\mathcal{O}(1)}$-time algorithm for problems such as \textsc{Independent Set}, \textsc{Vertex Cover}, and \textsc{3-Coloring}.
	
	However, until now, it was not known whether the \sdhd\ problem admits a polynomial kernel or not. 
	A \emph{kernelization} of a parameterized graph problem $\Pi$ is a polynomial-time algorithm which, given an instance $(G,k)$ of $\Pi$, 
	outputs an equivalent instance $(G',k')$ of $\Pi$ with $\abs{V(G')}+k'\leq h(k)$ for some computable function $h$. 
	The resulting instance $(G',k')$ of a kernelization is called a \emph{kernel}, and 
	in particular, when $h$ is a polynomial function, $\Pi$ is said to admit a \emph{polynomial kernel}.

\subsection*{Our Contribution}

In this paper, we show that \sdhd\ admits a polynomial kernel.

\begin{THM}\label{thm:main1}
\sdhd\ admits a polynomial kernel.
\end{THM}

	We find in Section~\ref{sec:approximation} an approximate \dhm\ with $\mathcal{O}(k^3\log n)$ vertices if the given instance is a \YES-instance.
	An important observation here is that in a distance-hereditary graph, there is a balanced separator, which is a complete bipartite subgraph (possibly containing edges in each part).
	By recursively finding such separators, 
	we will decompose the given graph into $D\uplus K_1 \uplus \cdots \uplus K_{\ell} \uplus X$, 
	where $\ell=\mathcal{O}(k\log n)$, $D$ is distance-hereditary, each $K_i$ is a complete bipartite subgraph, $\abs{X}=\mathcal{O}(k^3 \sqrt{\log k}\log n)$.
	We argue that if a graph $H$ is the disjoint union of a distance-hereditary graph and a complete bipartite graph,
	then in polynomial time, one can construct a DH-modulator of size $\mathcal{O}(k^2)$ in $H$ if $(H,k)$ is a \YES-instance.
	Using this sub-algorithm $\ell$ times, we will construct an approximate \dhm\ with $\mathcal{O}(k^3\log n)$ vertices. 
	This part follows a vein similar to the approach of Jansen and Pilipczuk~\cite{JansenP2016} for \textsc{Chordal Vertex Deletion}.
	Given a DH-modulator $S$ by adding $\mathcal{O}(k^2)$ per each vertex in $S$, 
	we will obtain in Section~\ref{sec:good} a new DH-modulator $S'$ of size $\mathcal{O}(k^5\log n)$ such that
	for every $v\in S'$, $G[(V(G)\setminus S)\cup \{v\}]$ is also distance-hereditary.
	We will call such a DH-modulator a \emph{good DH-modulator}.
		
	The remaining part is contributed to reduce the number of vertices in $G-S'$.	
	Note that distance-hereditary graphs may contain a large set of pairwise twins.
	In Section~\ref{sec:twinred}, we present a reduction rule that results in bounding the size of each set of pairwise twins in $G-S'$.
	We give in Section~\ref{sec:countcc} a reduction rule that results in bounding the number of connected components of $G-S'$.
	The last step is to reduce the size of each connected component of $G-S'$ having at least $2$ vertices.
	For this, we use split decompositions of distance-hereditary graphs.
	Briefly, split decompositions present tree-like structure of distance-hereditary graphs, 
	with a decomposition tree with bags for each nodes, such that
	each bag consists of a maximal set of pairwise twins in $G-S'$.
	Since the result of Section~\ref{sec:twinred} provides a bound of each maximal set of pairwise twins in $G-S'$, 
	it is sufficient to bound the number of bags in the decomposition tree.
	We summarize our algorithm in Section~\ref{sec:total}, and conclude with some further discussions in Section~\ref{sec:conclusion}.

\section{Preliminaries}\label{sec:prelim}

\smallskip

In this paper, all graphs are simple and finite. Given a graph $G$, we write the vertex set and edge set of $G$ as $V(G)$ and $E(G)$ respectively. Unless otherwise stated, we reserve $n$ to denote $\abs{V(G)}$. For a vertex $v$ of $G$, we denote by $G-v$ the graph obtained from $G$ by removing $v$ and all edges incident with it. For a vertex subset $S$ of $G$, we denote by $G-S$ the graph obtained by removing all vertices in $S$. For a vertex subset $S$ of $G$, let $G[S]$ be the subgraph of $G$ induced by $S$.
For a vertex $v$ in $G$, we denote by $N_G(v)$ the set of all neighbors of $v$ in $G$.
For a vertex subset $S$ of $G$, we denote by $N_G(S)$ the set of all vertices in $V(G)\setminus S$ that have a neighbor in $S$, 
and let $N_G[S]:=N_G(S)\cup S$. If the graph $G$ is clear from the context, then we may remove $G$ from the notation.
We say that a graph is \emph{trivial} if it consists of a single vertex, and \emph{non-trivial} otherwise.

For two vertex sets $A$ and $B$ in $G$, we say $A$ is \emph{complete} to $B$ if for every $v\in A$ and $w\in B$, $v$ is adjacent to $w$, and 
$A$ is \emph{anti-complete} to $B$ if for every $v\in A$ and $w\in B$, $v$ is not adjacent to $w$. A \emph{star} is a tree with a distinguished vertex, called the \emph{center}, adjacent to all other vertices. A \emph{complete graph} is a graph with all possible edges. 
Two vertices $v$ and $w$ of a graph $G$ are \emph{twins} if they have the same neighbors in $G-\{v,w\}$.
We say a vertex subset $S$ of $G$ is a \emph{twin set} of $G$ if the vertices in $S$ are pairwise twins in $G$.

A graph $H$ is a \emph{biclique} if there is a bipartition of $V(H)$ into non-empty sets $A\uplus B$ such that any two vertices $a\in A$ and $b\in B$ is adjacent. Notice that there may be edges among the vertices of $A$ or $B$. For a vertex subset $K\subseteq V(G)$, we say that $K$ is a biclique of $G$ if $G[K]$ is a biclique.

For a connected graph $G$, a vertex subset $S\subseteq V(G)$ is called a \emph{balanced vertex separator} of $G$ if every component of $G-S$ has at most $\frac{2}{3}\abs{V(G)}$ vertices. We allow $V(G)$ to be a trivial balanced vertex separator of $G$.
For a vertex subset $S$ of $G$, a path is called an \emph{$S$-path} if its end vertices are in $S$ and all other internal vertices are in $V(G)\setminus S$.

\smallskip

\subsection{Distance-hereditary graphs.} 
	A graph $G$ is \emph{distance-hereditary} if for every connected induced subgraph $H$ of $G$, 
	the distance between $u$ and $v$ in $H$ is the same as the distance between $u$ and $v$ in $G$. 
	The class of distance-hereditary graphs was introduced by Howorka~\cite{Howorka1977}, and attracted much attention 
	after the work of Bandelt and Murder~\cite{BandeltM1982} in 1982. 
	This graph class was characterized in a various way; for instance, distance-hereditary graphs are exactly the graphs 
	that can be constructed from a vertex by a sequence of adding twins or leaves~\cite{BandeltM1982},
	or these graphs are $(5,2)$-crossing chordal graphs meaning that every induced cycle of length at least $5$ contains two crossing chords~\cite{Howorka1977}.
	
	A graph is called a \emph{DH obstruction} if it is isomorphic to a gem, a house, a domino or an induced cycle of length at least 5, 
	that are depicted in Figure~\ref{fig:obsdh}. A DH obstruction is \emph{small} if it has at most $6$ vertices.

\begin{THM}[Bandelt and Mulder~\cite{BandeltM1982}]\label{thm:characterization}
A graph is distance-hereditary if and only if it has no induced subgraph isomorphic to one of DH obstructions.
\end{THM}

The following lemma from~\cite{KanteKKP15} is useful to find a DH obstruction.
\begin{LEM}\label{lem:createdhobs}
Given a graph $G$, if $P$  is an induced path in $G$ of length at least 3 and $v\in V(G)\setminus V(P)$ is adjacent with the end vertices of $P$, then $G[V(P)\cup \{v\}]$ contains a DH obstruction including $v$. 
\end{LEM}

\begin{figure}
\centerline{\includegraphics[scale=0.7]{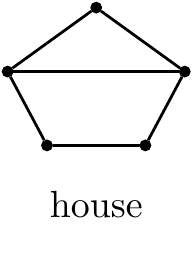} \quad\quad
\includegraphics[scale=0.7]{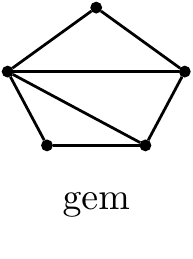} \quad\quad
\includegraphics[scale=0.7]{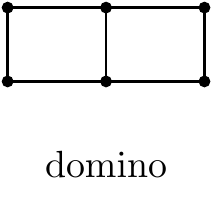} \quad\quad
\includegraphics[scale=0.7]{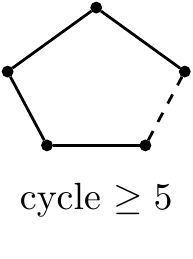} }
\caption{The induced subgraph obstructions for distance-hereditary graphs. }
\label{fig:obsdh}
\end{figure}

For a subset $S$ of a graph $G$, we say that $S$ is a \emph{\dhm} if $G-S$ is distance-hereditary. A \dhm\ $S$ is \emph{good} if every DH obstruction of $G$ contains at least two vertices of $S$, or equivalently, $G[(V(G)\setminus S)\cup \{v\}]$ is distance-hereditary for every $v\in S$.

\subsection{Split decompositions.}
We follow the notations in \cite{Bouchet1988a}.
A \emph{split} of a graph $G$ is a vertex partition $(A,B)$ of $G$ such that $\abs{A}\ge 2, \abs{B}\ge 2$, and $N_G(B)$ is complete to $N_G(A)$. 
A connected graph $G$ is called a \emph{prime graph} if $\abs{V(G)}\ge 5$ and it has no split.
A connected graph $D$ with a distinguished set of cut edges $M(D)$ of $D$ is called a \emph{marked graph} if
$M(D)$ forms a matching.
An edge in $M(D)$ is called a \emph{marked edge}, and every other edge is called an \emph{unmarked edge}.
A vertex incident with a marked edge is called a \emph{marked vertex},
and every other vertex is called an \emph{unmarked vertex}.
Each connected component of $D-M(D)$ is called a \emph{bag} of $D$.
See Figure~\ref{fig:decomposition} for an example.

\begin{figure}
\centerline{\includegraphics[scale=0.35]{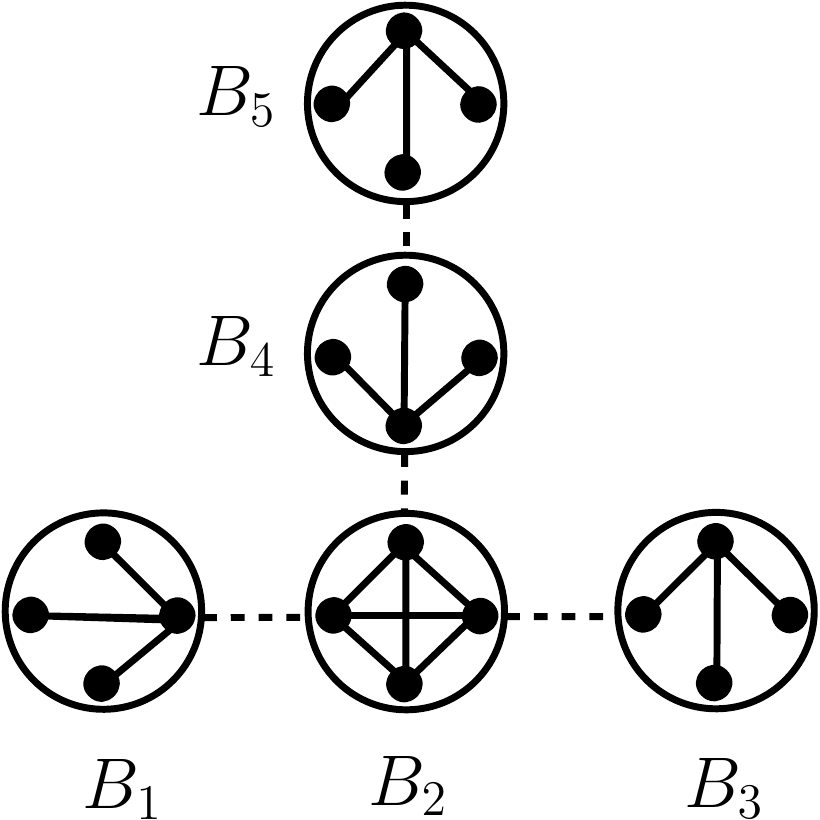} \quad\quad
\includegraphics[scale=0.35]{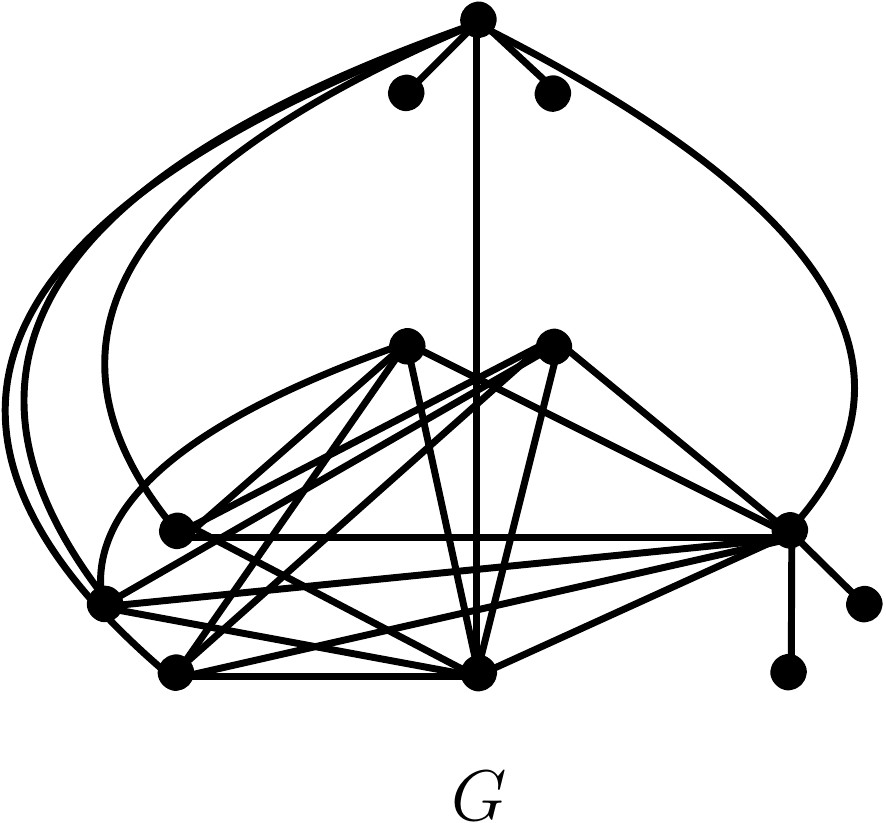} }
\caption{An example of a split decomposition of a distance-hereditary graph. Dashed edges denote marked edges and each $B_i$ denotes a bag. }
\label{fig:decomposition}
\end{figure}

When $G$ admits a split $(A, B)$, we construct a marked graph $D$ on the vertex set $V(G) \cup \{a',b'\}$ such that
\begin{itemize}
\item $a'b'$ is a new marked edge, 
\item $A$ is anti-complete to $B$, 
\item $\{a'\}$ is complete to $N_G(B)$, $\{b'\}$ is complete to $N_G(A)$, and 
\item for vertices $a,b$ with $\{a,b\}\subseteq A$ or $\{a,b\}\subseteq B$, $ab\in E(G)$ if and only if $ab\in E(D)$.
\end{itemize}
The marked graph $D$ is called a \emph{simple decomposition of} $G$.
A \emph{split decomposition} of a connected graph $G$ is a marked graph $D$ defined inductively to be either $G$ or a marked graph defined from a split decomposition $D'$
of $G$ by replacing a bag $B$ with its simple decomposition.  
It is known that for two vertices $u,v$ in $G$, $uv\in E(G)$ if and only if there is a path from $u$ to $v$ in $D$ where its first and last edges are unmarked, and
an unmarked edge and a marked edge alternatively appear in the path~\cite[Lemma 2.10]{AKK2014}. 
For convenience, we call a bag a \emph{star bag} or a \emph{complete bag} if it is a star or a complete graph, respectively.
 
Naturally, we can define a reverse operation of decomposing into a simple decomposition; for a marked edge $xy$ of a split decomposition $D$, 
\emph{recomposing $xy$} is the operation of removing two vertices $x$ and $y$ and making $N_D(x)\setminus \{y\}$ complete to $N_D(y)\setminus \{x\}$ with unmarked edges.
It is not hard to observe that if $D$ is a split decomposition of $G$, then $G$ can be obtained from $D$ by recomposing all marked edges.

Note that there are many ways of decomposing a complete graph or a star, because every its vertex partition $(A,B)$ with $\abs{A}\ge 2$ and $\abs{B}\ge 2$ is a split.
Cunningham and Edmonds \cite{CunninghamE80} developed a canonical way to decompose a graph into a split decomposition by not allowing to decompose a star bag or a complete bag.
A split decomposition $D$ of $G$ is called a \emph{canonical split decomposition} if each bag of $D$ is either a prime graph, a star, or a complete graph, and recomposing any marked edge of $D$ violates this property. Bouchet~\cite[(4.3)]{Bouchet1988a} observed that every canonical split decomposition has no marked edge linking two complete bags, and no marked edge linking a leaf of a star bag and the center of another star bag. Furthermore, for each pair of twins $a,b$ in $G$, it holds that $a,b$ must both be located in the same bag of the canonical split decomposition. 
\begin{THM}[Cunningham and Edmonds~\cite{CunninghamE80}] \label{thm:CED} 
Every connected graph has a unique canonical split decomposition, up to isomorphism.
\end{THM}
\begin{THM}[Dahlhaus~\cite{Dahlhaus00}]\label{thm:dahlhaus}
The canonical split decomposition of a graph  $G$ can be computed in time $\mathcal{O}(\abs{V(G)}+\abs{E(G)})$.
\end{THM}

We now give the second characterization of distance-hereditary graphs that is crucial for our results.

\begin{THM}[Bouchet~\cite{Bouchet1988a}]\label{thm:bouchet}
A graph is distance-hereditary if and only if every bag in its canonical split decomposition is either a star bag or a complete bag.
\end{THM}

\subsection{Extending a canonical split decomposition}\label{subsec:extendingsd}

Let $G$ and $H$ be connected graphs such that $G$ is distance-hereditary and $H$ is obtained from $G$ by adding a vertex $v$.
Gioan and Paul~\cite{GioanP2012} characterized when $H$ is again distance-hereditary or not, and described the way to extend the canonical split decomposition of $G$ to the canonical split decomposition of $H$
when $H$ is distance-hereditary. 
We need this characterization in Section~\ref{sec:boundingnontrivialcc}.

\begin{figure}
\centerline{
\includegraphics[scale=0.35]{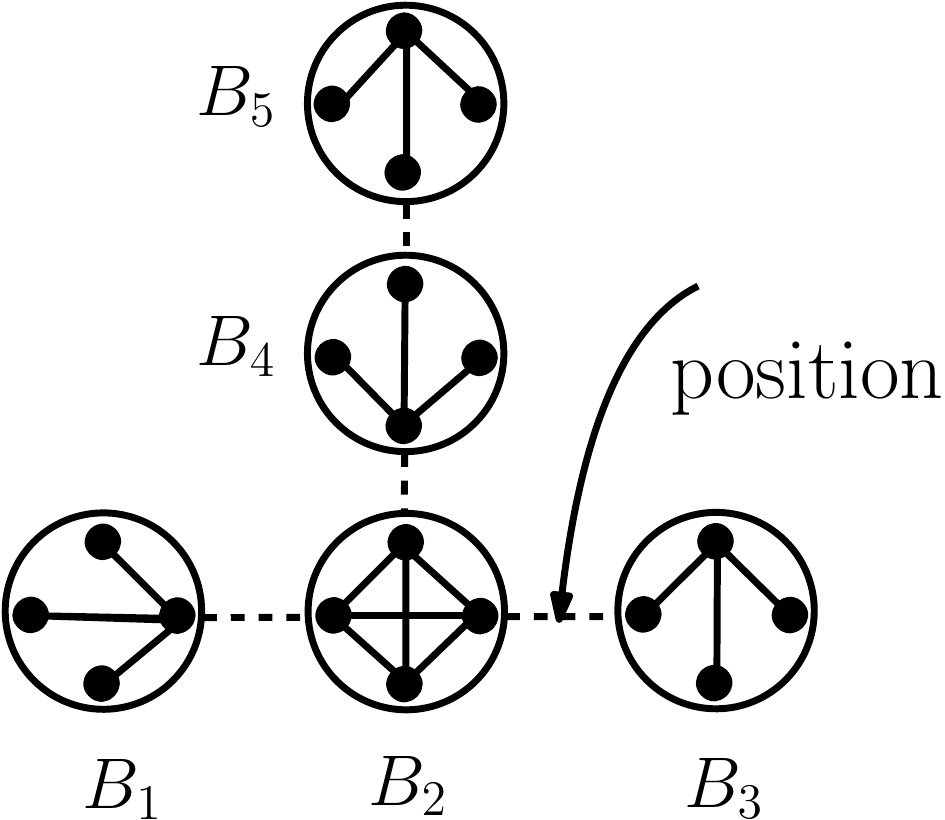} \quad\quad
\includegraphics[scale=0.35]{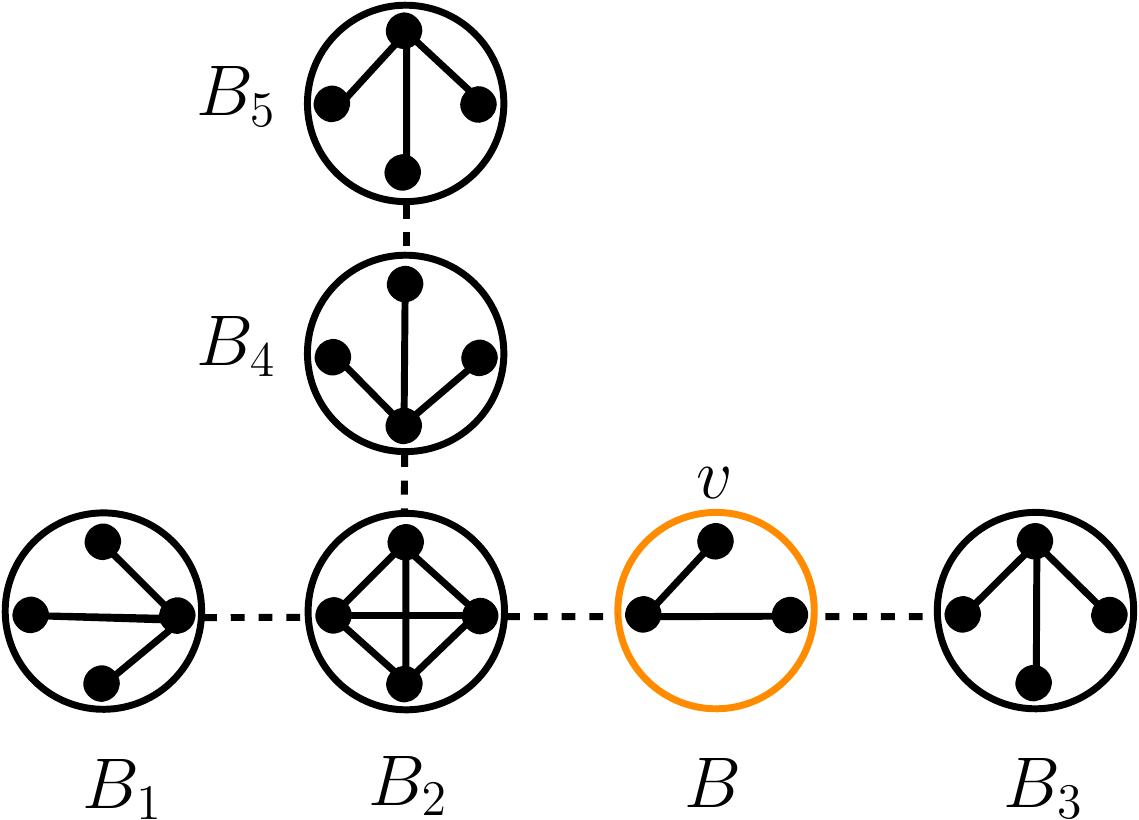} \quad\quad
\includegraphics[scale=0.35]{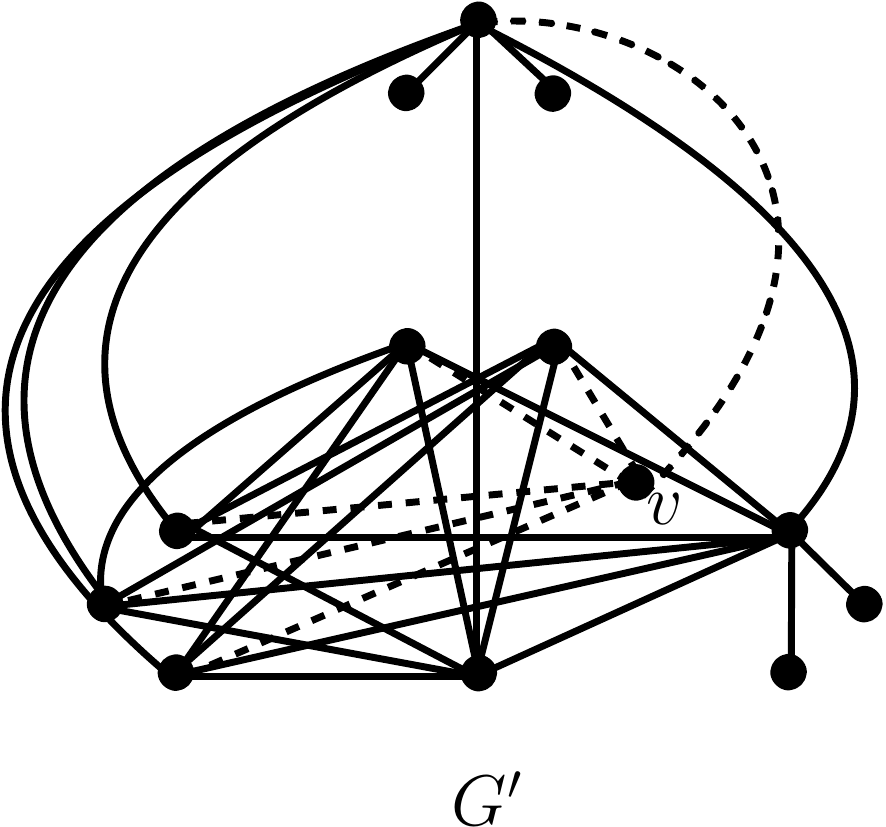} }
\caption{The way of extending the given split decomposition by adding a new vertex $v$. The pointed position is the marked edge described in (2) of Theorem~\ref{thm:GioanP2012}, and it can be computed in polynomial time. The second split decomposition is the modified split decomposition. }
\label{fig:dynamic}
\end{figure}

Let $D$ be the canonical split decomposition of a connected distance-hereditary graph $G$. For $S\subseteq V(G)$ and 
a vertex $v$ in $D$ with a bag $B$ containing $v$, $v$ is \emph{accessible} with respect to $S$ if either $v\in S$, or the component of $D-V(B)$ having a neighbor of $v$ contains a vertex in $S$. For $S\subseteq V(G)$ and a bag $B$ of $D$,
\begin{enumerate}
\item $B$ is \emph{fully accessible} with respect to $S$ if all vertices in $B$ are accessible with respect to $S$,
\item $B$ is \emph{singly accessible} with respect to $S$ if $B$ is a star bag of $D$, and exactly two vertices of $B$
including the center of $B$ are accessible with respect to $S$, and
\item $B$ is \emph{partially accessible} with respect to $S$ if otherwise.
\end{enumerate}
A star bag $B$ of $D$ is \emph{oriented towards} a bag $B'$ (or a marked edge $e$) in $D$ if the center of $B$ is marked, and  the path from the center of $B$ to a vertex of $B'$ (or to end vertices of $e$) contain the marked edge incident with the center of $B$. For $S\subseteq V(G)$, we define $D(S)$ as the minimal connected subdecomposition of $D$ such that
\begin{enumerate}
\item $D(S)$ is induced by the union of a set of bags of $D$, and
\item $D(S)$ contains all vertices of $S$.
\end{enumerate}

\begin{THM}[Gioan and Paul, Theorem 3.4 of \cite{GioanP2012}]\label{thm:GioanP2012}
Let $G$ and $H$ be connected graphs such that $G$ is distance-hereditary and $H$ is obtained from $G$ by adding a vertex $v$.
Let $N_H(v)=S$. Then $H$ is distance-hereditary if and only if
at most one bag of $D(S)$ is partially accessible in $D$, and
\begin{enumerate}[(1)]
\item if there is a partially accessible bag $B$ in $D(S)$, then each star bag $B'\neq B$ in $D(S)$ is oriented towards $B$ if and only if it is fully accessible,
\item otherwise, there exists a marked edge $e$ of $D(S)$ such that each star bag $B$ in $D(S)$ is oriented towards $e$ if and only if it is fully accessible.
\end{enumerate}
Furthermore, we can find the partially accessible bag $B$ in (1) or the marked edge $e$ in (2) in time $\mathcal{O}(\abs{V(G)})$. \footnote{Theorem 3.4 of \cite{GioanP2012} presented one more condition that every complete bag is either fully or partially accessible, which is redundant by definition.}
\end{THM}

In Section 4.1 of \cite{GioanP2012}, Gioan and Paul explained how to obtain a canonical split decomposition of $H$ from $D$ using 
such a partially accessible bag or a marked edge. 
Since it is sufficient to find such a bag or a marked edge for our purpose, we will not describe how to update it, 
but we give an example in Figure~\ref{fig:dynamic}.
One important property is that for every two vertices in a bag of $D$ that is not the partially accessible bag, 
they are still twins in $H$.

\section{Approximation algorithm}\label{sec:approximation}

We present a polynomial-time algorithm which constructs an approximate \dhm\ of $G$ whenever $(G,k)$ is a \YES-instance for \sdhd. The main result is as follows.

\begin{THM}\label{thm:approx}
There is a polynomial-time algorithm which, given a graph $G$ and a positive integer $k$, either correctly reports that $(G,k)$ is a \NO-instance to \sdhd, or returns a \dhm\ $S\subseteq V(G)$ of size $O(k^3\cdot \log n)$. 
\end{THM}

Recall that a DH obstruction is said to be small if it contains at most 6 vertices. If $G$ contains $k+1$ vertex-disjoint copies of small DH obstructions, then $(G,k)$ is clearly a \NO-instance. Therefore, we may assume that a maximal packing of small DH obstructions in $G$ has cardinality at most $k$. Notice that after removing all the vertices in a maximal packing, the resulting graph has no small DH obstruction. The following is the key statement for proving Theorem~\ref{thm:approx} and most part of this section is devoted to its proof.

\begin{PROP}\label{prop:modulator}
There is a polynomial-time algorithm which, given a graph $G$ without a small DH obstruction and a positive integer $k$, either correctly reports that $(G,k)$ is a \NO-instance to \sdhd, or returns a \dhm\ $S\subseteq V(G)$ of size $O(k^3\cdot \log n)$. 
\end{PROP}

A recent work of Jansen and Pilipczuk~\cite{JansenP2016} for a polynomial kernel for \textsc{Chordal Vertex Deletion} employs an approximation algorithm as an important subroutine. We follow a similar vein of approach as in~\cite{JansenP2016} here. The crux is to reduce the problem of finding an approximate \dhm\ to a restricted version where an input instance comes with a biclique \dhm\ (which we call a {\sl controlled instance}). In this special case, finding an approximate \dhm\ boils down to solving \textsc{Vertex Multicut}. We construct an instance of \textsc{Vertex Multicut} using a feasible solution to the LP relaxation of \sdhd. To this end, we obtain a decomposition 
\[V(G)=D\uplus K_{\ell} \uplus \cdots K_1 \uplus X_{\ell}\uplus \cdots X_1\] 
where $G[D]$ is distance-hereditary, each $K_i$ is a biclique, $\abs{X_i}=O(k\sqrt{\log{k}})$ and $\ell=O(k\cdot \log n)$, if $(G,k)$ is a \YES-instance (Proposition~\ref{prop:decomposition}).
Such a decomposition is achieved by  recursively extracting pairs $(K, X)$ from $G$, where $K$ is a biclique, $X$ is a vertex set of size at most $O(k\sqrt{\log k})$ until the resulting graph after removing those sets is distance-hereditary. Then in the graph $G[D\uplus K_{\ell} \uplus \cdots K_1 ]$, we recursively find an approximate DH-modulator using the algorithm for a controlled instance.

As we consider distance-hereditary graphs instead of chordal graphs, our approximation algorithm makes some important deviations from~\cite{JansenP2016}. First, we extract a pair containing a biclique, instead of a clique. Second, unlike in~\cite{JansenP2016}, we cannot guarantee that a long induced cycle traverses exactly one connected component of the distance-hereditary graph $D$. Such differences call for nontrivial tweaks in our approximation. 

Except for the last subsection in which we prove Theorem~\ref{thm:approx}, we assume that $G$ contains no small DH obstruction.
 
\subsection{LP relaxation and preprocessing}\label{subsec:preprocessing}

Given a graph $G$, let $x$ be a mapping from $V(G)$ to $\mathbb{R}$ and denote $x(v)$ as $x_v$ for every $v\in V(G)$. For a subgraph $H$ of $G$, we define $x(H):=\sum_{v\in V(H)}x_v$ and $\abs{x}:=x(G)$. As we assume that $G$ does not contain any small DH obstructions, the following is a linear program formulation of \sdhd\ for an instance $(G,k)$:
\begin{align*}
&\min \sum_{v} x_v \\
s.t \qquad & x(H) \geq 1 \quad & \forall \text{$H$ is an induced cycle of length at least 7}\\
&x_v \geq 0 \quad & \forall v\in V(G).
\end{align*} It is clear that an optimal integral solution corresponds to an optimal \dhm\ and vice versa. We call a feasible solution to this LP a \emph{feasible fractional solution} to \sdhd\ (for $G$).

\begin{OBS}\label{obs:subgraph}
If $x^*$ is a feasible fractional solution to \sdhd\ for a graph $G$, then for any induced subgraph $G'$ of $G$, $x^*$ restricted to $V(G')$ is a feasible fractional solution for $G'$ as well.
\end{OBS}

An optimal fractional solution $x^*$ to \sdhd\ can be found in polynomial time using the ellipsoid method\footnote{We mention that for our LP relaxation, found optimal solution has all rational values, each being represented using polynomial number of digits in $n$.} provided that there is a separation oracle for detecting a violated constraint for a fractional solution $x'$. Such a separation oracle can be easily implemented in polynomial time.

\begin{LEM}\label{lem:oracle}
Let $G$ be a graph with non-negative weights $x':V(G)\rightarrow \mathbb{Q}_{\geq 0}$ and $\ell\in \mathbb{Q}_{\geq 0}$. For any positive integer $d$, one can decide in time $O(n^{d+2})$ whether there is an induced cycle $H$ of length at least $d$ with $x'(H)< \ell$.
\end{LEM}
\begin{proof}
For every size-$d$ vertex subset $D$ of $G$ such that $G[D]$ is an induced path, with $s$ and $t$ as end vertices, we solve a weighted shortest-path problem on $G'$ with weights $x'|_{V(G')}$, where $G'$ is obtained by removing all internal vertices of $D$ and their neighbors, except for $s$ and $t$. Clearly, $D$ together with the obtained shortest $(s,t)$-path forms an induced cycle of length at least $d$. Conversely, if there is an induced cycle of length at least $d$, then it can be detected by considering all size-$d$ vertex subsets. 
\end{proof}

Let $x^*$ be an optimal fractional solution to \sdhd\ for $G$ and let $\tilde{X}$ be the set of all vertices $v$ such that $x^*_v\geq \frac{1}{20}$. Observe  
\begin{equation}
k\geq x^*(G)\geq \frac{1}{20}\cdot \abs{\tilde{X}} + x^*(G-\tilde{X}) 
\end{equation}
From Inequality (1), it is easy to see that if $\abs{\tilde{X}}> 20k$, then $(G,k)$ is a \NO-instance. Therefore, we may assume that $\abs{\tilde{X}}\le 20k$. Furthermore, by  Observation~\ref{obs:subgraph}, $x^*$ restricted to $V(G)\setminus \tilde{X}$ is a feasible fractional solution for $G-\tilde{X}$ and $x^*_v<\frac{1}{20}$ for every $v\in V(G)\setminus \tilde{X}$. We add $\tilde{X}$ to the approximate \dhm\ $S$ we are constructing and remove $\tilde{X}$ from $G$. Henceforth, we assume that $x^*$ is a feasible fractional solution to \sdhd\ for $G$ such that $x^*_v<\frac{1}{20}$ for every $v\in V(G)$. Clearly, $\abs{x^*}\leq k$.

\subsection{Decomposition into DH and bicliques} \label{subsec:decomposition}

Using the following observations, we decompose the vertex set of $G$ into bicliques, a set inducing a distance-hereditary graph, and a bounded number of extra vertices.

We first prove an $\mathcal{O}(n^3)$ bound on the number of maximal bicliques in a graph having no small DH obstructions.
This can be seen as an extension of an $\mathcal{O}(n^2)$ bound on the number of maximal cliques in graphs without $C_4$~\cite[Proposition 2]{Farber1989} or an $\mathcal{O}(n\sqrt{n})$ bound for graphs without $C_4$ and the diamond~\cite[Theorem 2.1]{EschenHCSS2011}. The tightness of the latter result was also confirmed.
We say two sets $A$ and $B$ \emph{cross} if $A\setminus B\neq \emptyset$, $A\cap B\neq \emptyset$, and $B\setminus A\neq \emptyset$.

\begin{LEM}\label{lem:noncrossingpartition}
Let $A$ be a set and let $\mathcal{C}=\{C_1, \ldots, C_m\}$ be a family of subsets of $A$ such that 
for each distinct integers $i, j\in \{1, \ldots, m\}$,  $C_i\neq C_j$ and $C_i$ and $C_j$ do not cross.
Then $m\le \frac{\abs{A}(\abs{A}+1)}{2}$.
\end{LEM}
\begin{proof}
We may assume $\abs{A}\ge 2$. Let $v\in A$. Note that the sets in $\mathcal{C}$ containing $v$ can be ordered linearly by the inclusion relation.
Therefore, there are at most $\abs{A}$ sets containing $v$.
By induction $\{C_i\in \mathcal{C} : v\notin C_i\}$ contains at most $\frac{(\abs{A}-1)\abs{A}}{2}$ elements.
Thus, in total, $\abs{\mathcal{C}}\le \frac{\abs{A}(\abs{A}+1)}{2}$.
\end{proof}

\begin{LEM}\label{lem:numbiclique}
Let $G$ be a graph on $n$ vertices that has no small DH obstructions. Then $G$ contains at most $\frac{n^3+5n}{6}$ maximal bicliques, and they can be enumerated in polynomial time.
\end{LEM}
\begin{proof}
We claim that every vertex is contained at most $\frac{(n-1)n}{2}+1$ maximal bicliques of $G$. 
Then by recursively enumerating all maximal bicliques containing $v$ of $G$ and discarding $v$ from $G$, we obtain the bound $\sum_{1\le i\le n} \left( \frac{(i-1)i}{2}+1 \right)\le \frac{n^3+5n}{6}$ on the number of maximal bicliques in total.

Let $v\in V(G)$, and let $N_1:=N_G(v)$ and let $N_2$ be the set of vertices in $G- (\{v\}\cup N_G(v))$ that have a neighbor in $N_1$.
We claim that for $w,z\in N_2$, $N_G(w)\cap N_1$ and $N_G(z)\cap N_1$ do not cross.
Suppose for contradiction that there exists $w, z\in N_2$ where $N_G(w)\cap N_1$ and $N_G(z)\cap N_1$ cross.
By definition, there exists $a_1\in (N_G(w)\setminus N_G(z))\cap N_1$, $a_2\in  (N_G(w)\cap N_G(z))\cap N_1$,  and $a_3\in (N_G(z)\setminus N_G(w))\cap N_1$.
If $wz\in E(G)$, then regardless of the adjacency between $a_1$ and $a_3$, 
$G[\{w, z, a_1, a_3, v\}]$ is isomorphic to $C_5$ or the house.
We may assume $wz\notin E(G)$.
If $a_1a_3\in E(G)$, then $G[\{w, z, a_1, a_2, a_3\}]$ is isomorphic to $C_5$, the house, or the gem.
If $a_1a_3\notin E(G)$, then $G[\{v,w,z,a_1, a_2, a_3\}]$ contains an induced subgraph isomorphic to the domino, the house, or the gem.
Therefore, we can always find a small DH obstruction.
We conclude that for $w,z\in N_2$, $N_G(w)\cap N_1$ and $N_G(z)\cap N_1$ do not cross, as $G$ has no small DH obstructions.

So, $\{N_G(v)\cap N_1:v\in N_2\}$ forms a family of sets such that any two sets do not cross.
Since $\abs{N_1}\le \abs{V(G)}-1$, by Lemma~\ref{lem:noncrossingpartition}, it contains at most $\frac{(n-1)n}{2}$ sets.

Notice that every biclique containing $v$ and a vertex in $N_2$ should consists of one part with the union of $v$ and a set of vertices in $N_2$ having the same neighbors on $N_1$.
Thus, the number of such bicliques is at most $ \frac{(n-1)n}{2}$. With the maximal biclique containing the trivial biclique with the partitions $\{v\}$ and $N_1$, 
there are at most $\frac{(n-1)n}{2}+1$ bicliques containing $v$, as required.
Furthermore, by classfying the classes of vertices in $N_2$ having the same neighborhoods in $N_1$, 
we can enumerate all bicliques in polynomial time.
\end{proof}

One main observation is that every connected distance-hereditary graph admits a biclique that is also a balanced vertex separator.
We prove this.

\begin{LEM}\label{lem:sepbiclique}
Let $G$ be a connected distance-hereditary graph on at least two vertices. Then $G$ contains a blicique that is a balanced vertex separator of $G$. 
\end{LEM}

\begin{proof}%
Let $D$ be a canonical split decomposition of $G$. If $D$ consists of a single bag $B$, then $B=G$. 
Note that $B$ is either a star or a complete graph on at least two vertices. In this case, we take the whole vertex set $V(G)$ as a balanced vertex separator of $G$. 
Henceforth, we can assume that $D$ contains at least two bags.

For every marked edge $e=xy$ of $D$, let $U_x(e)$ (respectively, $U_y(e)$) denote the set of unmarked vertices contained in the connected component of $D-e$ containing $x$ (respectively, $y$). Recall that $(U_x(e),U_y(e))$ defines a split of $G$. Let $R_x(e):=N_G(U_y(e))$, $R_y(e):=N_G(U_x(e))$, and $K_e:=R_x(e)\cup R_y(e)$.
We note that a vertex $v$ belongs to $R_y(e)$ if and only if $v$ is an unmarked vertex in $D$ that is connected to $y$ by an odd-length path in $D-e$ alternating in marked edges and unmarked edges. Since $(U_x(e),U_y(e))$ is a split, the vertex set  $K_e$ forms a biclique. 

Let  $\vec{D}$ be a graph obtained by orienting every marked edge $e=xy$ from $x$ to $y$ whenever $G[U_y(e)\setminus R_y(e)]$ contains a connected component of size larger than $\frac{1}{2}n$. Note that a connected component of $G[U_y(e)\setminus R_y(e)]$ is a connected component of $G-K_e$. If there is an unoriented edge $e=xy$ in $\vec{D}$,  $K:=K_e$ is a desired biclique that is a balanced vertex separator. Hence we assume that every marked edge is oriented in $\vec{D}$. Since every bag of $\vec{D}$ can have at most one marked edge oriented outward, there exists a bag $B$ such that all incident marked edges are oriented toward $B$  in $\vec{D}$. 

Consider the case when $B$ is a leaf bag, $e=xy$ is the unique marked edge incident with $B$ oriented from $x$ to $y$, and $y\in V(B)$. If $B$ is a complete bag or a star bag in which $y$ is the center, clearly all unmarked vertices of $B$ belong to $R_y(e)$, and thus $U_y(e)\setminus R_y(e)=\emptyset$, contradicting the orientation of $e$. If $B$ is a star bag in which $y$ is a leaf, the (unmarked) center of $B$ belongs to $R_y(e)$. Observe that each unmarked leaf vertex of $B$ forms a trivial connected component of $G[U_y(e)\setminus R_y(e)]$, contradicting the orientation of $e$.

Hence $B$ is incident with at least two marked edges, say $e_i=x_iy_i$ for $i=1,2,\ldots $, where $y_i\in V(B)$ for every $i$. Let $U_B$ be the set of unmarked vertices in $B$. We argue that $B$ is a star bag with its center unmarked.

\medskip
\noindent (a) \emph{$B$ is a complete bag:} Observe that for every $i$, $R_{y_i}(e_i)=\bigcup_{j\neq i} R_{x_j}(e_j)\cup U_B$. This means $K_{e_i}=R_{x_i}(e_i)\cup R_{y_i}(e_i)=\bigcup_{j} R_{x_j}(e_j)\cup B$, that is, $K_{e_i}$ is the same for every $i$. Moreover, we have $U_B\subseteq K_{e_i}$. This means that each connected component of $G-K_{e_i}$ of size at most $\frac{1}{2}n$, contradicting the orientations.

\medskip
\noindent (b) \emph{$B$ is a star bag with its center marked:} Notice that the center of $B$ is one of $y_i$'s, say $y_1$. The orientation of $e_1$ implies that $G[U_{y_1}(e_1)\setminus R_{y_1}(e_1)]$ has a connected component $C$ containing more than $\frac{1}{2}n$ vertices. On the other hand, all vertices of $U_B$ and $R_{x_i}(e_i)$ for $i\neq 1$ have an odd-length path in $D-e_1$ to $y_1$ alternating in marked and unmarked edges, implying that $U_B\cup \bigcup_{i\neq 1}R_{x_i}(e_i)\subseteq R_{y_1}(e_1)$. Therefore, $C$ is contained in a connected component of $G[U_{x_j}(e_j)\setminus R_{x_j}(e_j)]$ for some $j\neq 1$. This contradicts the orientation of $e_j$.

\medskip
Hence, $B$ is a star bag with its center unmarked. Let $v$ be the center of $B$ and observe that for every $i$, $v$ is the only vertex in $B_{y_i}(e_i)$. Observe that $K:=\{v\}\cup \bigcup_i R_{x_i}(e_i)$ is a biclique with the desired property. This completes the proof.
\end{proof}

The following result is from~\cite{FeigeHL08}.
\begin{THM}[Feige, Hajiaghayi, and Lee~\cite{FeigeHL08}]\label{thm:approxsep}
There is an $O(\sqrt{\log{\sf opt}})$-approximation algorithm for finding a balanced vertex separator.
\end{THM}

\begin{LEM}\label{lem:findsep}
Let $(G,k)$ be an instance to \sdhd\ such that $G$  is connected and contains no small DH obstructions. There is a polynomial-time algorithm which finds a balanced vertex separator $K\uplus X$ such that
\begin{itemize}
\item[-]  $K$ is a biclique or an empty set,
\item[-] and $\abs{X}=O(k\sqrt{\log{k}})$
\end{itemize}
whenever $(G,k)$ is a \YES-intance.
\end{LEM}
\begin{proof}
We perform the following algorithm. Let $K$ be a maximal biclique of $G$ or an empty set.
If every connected component of $G-K$ contains at most $\frac{2}{3}n$ vertices, we take $X:=\emptyset$. Otherwise, we apply Theorem~\ref{thm:approxsep} to a (unique) largest connected component $C$ of $G-K$ to find a balanced vertex separator $X$ of $C$ of size $O({\sf opt}\sqrt{\log{{\sf opt}}})$. In both cases, if $X$ is of size at most $O(k\sqrt{\log{k}})$, then $K\uplus X$ is clearly a desired balanced vertex separator. If no balanced vertex separator of size at most $O(k\sqrt{\log{k}})$ is found while iterating over all maximal bicliques of $G$ and an empty set as $K$, then we report that there is no \dhm\ of size at most $k$. This algorithm runs in polynomial time since a balanced vertex separator $X$ of a connected component $C$ can be found in polynomial time by Theorem~\ref{thm:approxsep}, and Lemma~\ref{lem:numbiclique} provides an efficient way to iterate over all maximal bicliques. 

To see the correctness, suppose that $(G,k)$ is a \YES-instance and let $X_0$ be an optimal \dhm. If $G-X_0$ forms an independent set, then $X_0$ is a balanced vertex separator of $G-K$, where $K:=\emptyset$. Therefore, the approximation algorithm of Theorem~\ref{thm:approxsep} will indeed find $O(k\sqrt{\log{k}})$-size balanced vertex separator when the above algorithm considers $K=\emptyset$.  

Suppose this is not the case, and let $C$ be a largest connected component of $G-X_0$. By Lemma~\ref{lem:sepbiclique}, there is a balanced vertex separator $K_0$ of $C$ which is a biclique. Let $K$ be a maximal biclique in $G$ with $K_0\subseteq K$. Since each connected component of $G-X_0-K_0$, and thus of $G-X_0-K$, contains at most $\frac{2}{3}\abs{C}\leq \frac{2}{3}n$ vertices, $X_0\cup K$ is a balanced vertex separator of $G$. Let $C'$ be any largest connected component of $G-K$. If $C'$ contains at most $\frac{2}{3}n$ vertices, then the above algorithm will set $X:=\emptyset$. If not, observe that $(X_0\setminus K)\cap C'$ is a balanced vertex separator of $C'$ whose size is at most $k$. Hence, the above algorithm applies the approximation algorithm of Theorem~\ref{thm:approxsep}  to $C'$, and   finds $O(k\sqrt{\log{k}})$-size balanced vertex separator $X'$ of $C'$. Notice that each connected component of $G-K-X'$ has size at most $\max\{\frac{1}{3}n,\frac{2}{3}\abs{C'}\}\leq \frac{2}{3}n$. This completes the proof.
\end{proof}

For an instance $(G,k)$ such that $G$ does not contain any small DH obstructions, let $K_1\uplus X_1$ be a balanced vertex separator obtained by applying Lemma~\ref{lem:findsep} to a connected component of $G$ that is not distance-hereditary. Notice that $K_1\uplus X_1$ is not necessarily a \dhm\ and there may be a connected component of $G-(K_1\uplus X_1)$ which is not distance-hereditary. At $i$-th recursive step, we apply Lemma~\ref{lem:findsep} to a connected component $G_i$ of $G-\bigcup_{j< i} (K_j\uplus X_j)$ which is not distance-hereditary and  obtain a balanced vertex separator $K_i\uplus X_i$ of $G_i$. If the algorithm of Lemma~\ref{lem:findsep} reports that $(G_i,k)$ is a \NO-instance for some $i$, then indeed $(G,k)$ is a \NO-instance. Otherwise, we obtain a decomposition  \[V(G)=D\uplus K_{\ell} \uplus \cdots K_1 \uplus X_{\ell} \uplus \cdots X_1,\] where $G[D]$ is distance-hereditary, each $K_i$ is a biclique or an empty set, and $\abs{X_i}=O(k\sqrt{\log{k}})$. 

The recursive applications of Lemma~\ref{lem:findsep} can be represented as a collection of branching trees $\mathcal{T}$, 
where each internal node corresponds  to the initial connected component (that are not distance-hereditary) to Lemma~\ref{lem:findsep} and its children correspond to 
the connected components obtained after removing a balanced separator. %
Suppose that $(G,k)$ is a \YES-instance and let $S$ be a size-$k$ modulator of $G$. Observe that each connected component corresponding to an internal node of $\mathcal{T}$ contains at least one vertex of $S$. Since the maximum length of a root-to-leaf path in $\mathcal{T}$ is $O(\log n)$, any vertex of $G$ (and thus, $S$) appears in at most $O(\log n)$ connected components represented as the nodes of $\mathcal{T}$. It follows that the number of internal nodes of $\mathcal{T}$ is at most $O(k\log n)$, and thus $\ell =O(k\log{n})$. This implies $\bigcup_{i\leq \ell} X_i$ has at most $O(k^2\sqrt{\log{k}}\log n)$ vertices. We add them to a \dhm\ $S$ under construction and remove from $G$.

We summarize the result of this subsection.

\begin{PROP}\label{prop:decomposition}
Let $(G,k)$ be an instance to \sdhd\ such that $G$ contains no small DH obstructions. There is a polynomial-time algorithm which, given such $(G,k)$, either computes a decomposition  \[V(G)=D\uplus K_{\ell} \uplus \cdots K_1 \uplus X\] 
such that $G[D]$ is distance-hereditary, each $K_i$ is a biclique, $\abs{X}=O(k^2\sqrt{\log{k}}\log n)$ and $\ell=O(k\log n)$, or correctly reports that $(G,k)$ is a \NO-instance.
\end{PROP}

\subsection{Handling a controlled instance}

A graph $G$ is called a \emph{controlled graph with a partition $(D,K)$} if $G$ contains no small DH obstructions
and $V(G)$ is partitioned into $D\uplus K$ such that $G[D]$ is distance-hereditary and $K$ is a biclique. We fix such a partition $D\uplus K$  for a given controlled instance under consideration. Also we fix a bipartition $A\uplus B$ for the given biclique $K$ such that $ab\in E(G)$ for every $a\in A$ and $b\in B$. In this part, it is important that both $A$ and $B$ are non-empty by the definition of a biclique.

The next lemma is useful. 

\begin{LEM}\label{lem:longpath}
Let $G$ be a controlled graph with a partition $(D,K)$. For an induced cycle $H$ of length at least 7, let $\ell$ be the number of connected components in $G[V(H)\cap K]$. Then the followings hold. 
\begin{enumerate}
\item[(a)] If $\ell=1$, then both $G[V(H)\cap K]$ and $H-K$ are paths, and $\abs{V(H)\cap K}\leq 3$.
\item[(b)] If $\ell\geq 2$, 
there exists a subpath $P$ of $H$ such that $P$ is a $K$-path, and the length of $P$ is at least 3.
\end{enumerate}
\end{LEM}
\begin{proof}
If $\ell=1$, clearly both $G[V(H)\cap K]$ and $H-K$ are (induced) paths. Suppose $\abs{V(H)\cap K}\geq 4$, and note that $V(H)\cap K$ are contained  in either $A$ or $B$, say $A$. Then $G[(V(H)\cap K)\cup \{b\}]$ induces a gem for any $b\in B$ since $G[V(H)\cap K]$ is an induced path. This contradicts that $G$ does not contain any small DH obstructions. It follows $\abs{V(H)\cap K}\leq 3$.

Suppose $\ell\geq 2$. There exist two $K$-paths on $H$, say $P_1$ and $P_2$, each containing at least one internal vertex and occuring on $H$  consecutively. (Recall that an $S$-path is a path whose end vertices are in $S$ and all of whose internal vertices lie outside $S$.) We note that the end vertices of both $P_1$ and $P_2$ are contained  in either $A$ or $B$, say $A$; otherwise $H$ contains a chord. To prove (b), it suffices to show that one of $P_1$ and $P_2$ is of length at least three. Obviously, we have $\abs{V(P_i)}\geq 3$ for $i=1,2$. Suppose that  $\abs{V(P_1)}=\abs{V(P_{2})}=3$. Let $Q$ be a component of $G[V(H)\cap K]$ that intersects with both $P_1$ and $P_2$ (possibly $Q$ consists of a single vertex). Pick an arbitrary vertex $b\in B$; such $b$ exists for $B\neq \emptyset$. If $\abs{V(Q)}=1$, observe that $G[V(P_1)\cup V(P_{2})\cup \{b\}]$ induces a gem, a house or a domino, a contradiction. If $\abs{V(Q)}\geq 2$, then $G[V(P_1)\cup V(Q)\cup \{b\}]$ contains a gem or a house, a contradiction. This establishes (b).
\end{proof}

\begin{LEM}\label{lem:newDHDsol}
Let $G$ be a controlled graph with a partition $(D, K)$. If $x^*$ is a fractional solution to \sdhd\ for $G$ such that $x^*_v<\frac{1}{20}$ for every vertex $v$ of $G$, then $x'$ defined as 
\begin{align*}
x'_v=
\begin{cases}
0 & \text{if } v\in K\\
2x^*_v &\text{if } v\in D
\end{cases}
\end{align*}
is also a feasible fractional solution.
\end{LEM}
\begin{proof}
Consider an arbitrary induced cycle $H$ of length at least 7. We only need to verify that $x'(H) \geq 1$. 
There are two possibilities.

Suppose $G[V(H)\cap K]$ contains exactly one connected component. By Lemma~\ref{lem:longpath}, we have $\abs{V(H)\cap K}\leq 3$. By assumption, we have $\sum_{v\in V(H)\cap K} x^*_v < 3\cdot \frac{1}{20} <\frac{1}{2}$, and thus $\sum_{v\in V(H)\cap D} x^*_v \geq \frac{1}{2}$. It follows that \[ \sum_{v\in V(H)}x'_v=\sum_{v\in V(H)\cap D} x'_v =\sum_{v\in V(H)\cap D} 2x^*_v \geq 1.\]

If $G[V(H)\cap K]$ contains at least two connected components, by Lemma~\ref{lem:longpath}, $H$ contains a $K$-path $P$ of length at least 3 as a subpath. Recall that $V(H)\cap K$ are contained  in either $A$ or $B$, say $A$. Hence, for any $b\in B$, $G[\{b\}\cup V(P)]$ contains an induced cycle $H'$ of length at least 7 by Lemma~\ref{lem:createdhobs}. Note that $G[V(H')\cap K]$ has exactly one connected component and thus $x'(H')\geq 1$. Since $x'_b=0$, we have \[ \sum_{v\in V(H)} x'_v\geq \sum_{v\in V(P)} x'_v \geq \sum_{v\in V(H')}x'_v\geq 1.\qedhere \] \end{proof}
We mention that the feasible fractional solution $x'$ obtained as in Lemma~\ref{lem:newDHDsol} meets $\abs{x'}\leq 2\abs{x^*}$ and for every vertex $v$ of $G$, we have $x'_v<\frac{1}{10}$.

An instance of \textsc{Vertex Multicut} consists of an undirected graph $G$ and a set $\mathcal{T}$ of (unordered) vertex pairs of $G$. The goal is to find a minimum-size set $X\subseteq V(G)$ that hits every path from $s$ to $t$ for every $(s,t)\in \mathcal{T}$. Notice that we are allowed to delete a terminal. An LP formulation of $(G,\mathcal{T})$ of \textsc{Vertex Multicut} is the same as the above LP of \sdhd, except that we replace the constraints on induced cycles of length at least 7 by \[x(P)\geq 1\] for every $(s,t)$-path $P$ with $(s,t)\in \mathcal{T}$. 
The following result from~\cite{Gupta03} is originally stated for \textsc{Directed Multicut}, where we want to find a minimum set of arcs to hit all directed $(s,t)$-path for every (ordered) terminal pairs. Using standard reductions, one can reduce the undirected version of \textsc{Vertex Multicut} to \textsc{Directed Vertex Multicut}, which again can be reduced to  \textsc{Directed Multicut}. Moreover, these reductions preserve the objective values of both the integral solutions and feasible fractional solutions.

\begin{THM}[Gupta~\cite{Gupta03}]\label{thm:fractionalmulticut}
Let $x'$ be a feasible fractional solution to a \textsc{Vertex Multicut} instance $(G,\mathcal{T})$. There exists a constant $c$ such that, in polynomial time, one can find an integral solution to \textsc{Vertex Multicut} of size at most $c \cdot \abs{x'}^2$.
\end{THM}

\begin{LEM}\label{lem:usemulticut}
Let $G$ be a controlled graph with a partition $(D,K)$ and let $x'$ be a feasible fractional solution to \sdhd\  such that $x'_v< \frac{1}{10}$ for all $v\in D$ and $x'_v=0$ for all $v\in K$. There is a polynomial-time algorithm which, given such $(G,k)$ and $x'$, returns a \dhm\ $X$ of size $O(\abs{x'}^2)$.
\end{LEM}
\begin{proof}
We construct an instance $(G[D],\mathcal T)$ of \textsc{Vertex Multicut} with terminal pairs $\mathcal{T}:=\{(s,t)\subseteq D\times D:\dist_{G[D],x'}(s,t)\geq 1\}$, 
where $\dist_{G[D],x'}(s,t)$ is the minimum $x'(P)$ over all $(s,t)$-paths $P$. Notice that for every terminal pair $(s,t)\in \mathcal{T}$, and for every $(s,t)$-path $P$ in $G[D]$, we have \[x'(P) \geq \dist_{G[D],x'}(s,t)\geq 1,\] meaning that $x'$ is a feasible fractional solution to \textsc{Vertex Multicut} for the instance $(G[D],\mathcal{T})$. By Theorem~\ref{thm:fractionalmulticut}, we can obtain a vertex set $X\subseteq D$ of size $O( \abs{x'}^2)$ such that $G[D\setminus X]$ contains no $(s,t)$-path for every terminal pair $(s,t)\in \mathcal{T}$ in polynomial time.

It is sufficient to show that $G-X$ is distance-hereditary. For the sake of contradiction, suppose $G-X$ contains an induced cycle $H$ of length at least 7, and we specifically choose $H$ so as to minimize 
\begin{center}
the number of connected components in $G[V(H)\cap K]$
\end{center}

\begin{CLAIM}\label{clm:onecc}
$G[V(H)\cap K]$ has exactly one connected component and $\abs{V(H)\cap K}\leq 3$.
\end{CLAIM}
\begin{proofofclaim}
Suppose $G[V(H)\cap K]$ has at least two connected components. Observe that $V(H)\cap K$ is entirely contained either in $A$ or $B$, say $A$. By Lemma~\ref{lem:longpath}, there exists a subpath $P$ of $H$ that is a $K$-path and having length at least 3. Choose any vertex $b\in B$ and  observe that $G[V(P)\cup \{b\}]$ contains a DH obstruction $H'$ by Lemma~\ref{lem:createdhobs}. Especially, $H'$ must be an induced cycle of length at least 7 since $G$ does not contain any small DH obstructions. For $B\cap X=\emptyset$, $H'$ is an induced cycle in $G-X$ having strictly less connected components in $G[V(H')\cap K]$, contradicting the choice of $H$. This proves the first part of the statement. The second part follows from Lemma~\ref{lem:longpath}.
\end{proofofclaim}

By Claim~\ref{clm:onecc}, $H-K$ is an induced path $P$ contained in a connected component of $G[D]$. Let $s$ and $t$ be the end vertices of $P$ and notice that $(s,t)\notin \mathcal{T}$ since $P$ contains no vertex of $X$. Hence, $G[D]$ contains an induced $(s,t)$-path $W$ such that $x'(W)<1$. Since $x'$ is a feasible fractional solution to \sdhd, we have $x'(P)\ge x'(H)\ge 1$. Therefore, $P$ contains at least 11 vertices since $x'_v<\frac{1}{10}$ for every $v\in D$ and $x'_v=0$ for every $v\in K$. Since $G[D]$ is distance-hereditary and $W$ is induced, we have $\abs{V(W)}=\abs{V(P)}\ge 11$. Let $w_1(=s),w_1,\ldots , w_p(=t)$ be the vertices of $W$ in the order of their occurrence on $W$ where $p\ge 11$.

Now we argue that $G[V(W)\cup (V(H)\cap K)]$ contains a DH obstruction. 

\begin{CLAIM}\label{clm:hole7}
$G[V(W)\cup (V(H)\cap K)]$ contains an induced cycle of length at least 7. 
\end{CLAIM}
\begin{proofofclaim}
If $\abs{V(H)\cap K}=1$, this follows from Lemma~\ref{lem:createdhobs}. So, we may assume that $2\leq \abs{V(H)\cap K}\leq 3$. Let $v_1$, $v_2$ be vertices of $V(H)\cap K$ that are adjacent with $s=w_1$ and $t=w_p$, respectively. 

If $v_1$ is adjacent with any of $w_j$ with $j\geq 4$, then Lemma~\ref{lem:createdhobs} applies and we have a DH obstruction. Such an obstruction can only be an induced cycle of length at least 7 as we assume that $G$ does not contain a small DH obstruction. Hence, we may assume that $v_1$ is not adjacent with any of $w_4,\ldots , w_{p}.$
By a symmetric argument, we may assume that $v_2$ is not adjacent with any of $w_1,\ldots , w_{p-3}$.

Let $Q$ be a shortest path between $v_1$ and $v_2$ such that the set of internal vertices of $Q$ is non-empty and  contained $\{w_1,\ldots , w_p\}$. Such $Q$ exists since $v_1$ and $v_2$ are adjacent with $w_1$ and $w_p$, respectively. By the assumption that $v_1$ (respectively, $v_2$) is not adjacent with $w_j$ for $4\leq j\leq p$ (respectively, $1\leq j\leq p-3$),  the length of $Q$ clearly exceeds $6$. If $v_1$ and $v_2$ are the only vertices of $V(H)\cap K$, then $v_1v_2\in E(G)$ and observe that $G[\{v_1,v_2\}\cup V(Q)]$ forms an induced cycle of length at least 7. If there is another vertex in $V(H)\cap K$, observe that the remaining vertex, say $v_3$, is adjacent to $v_1$ and $v_2$. By Lemma~\ref{lem:createdhobs}, there exists a DH obstruction in $G[\{v_1,v_2,v_3\}\cup V(Q)]$, which can only be an induced cycle of length at least 7. 
\end{proofofclaim}

Let $\tilde{H}$ be such an induced cycle of length at least 7 as stated in Claim~\ref{clm:hole7}.  That $V(\tilde{H})\subseteq V(W)\cup (V(H)\cap K)$ and $x'(V(H)\cap K)=0$ implies $x'(\tilde{H})\leq x'(W)<1$. This contradicts the assumption that $x'$ is a feasible fractional solution to \sdhd. This concludes the proof that $G-X$ is distance-hereditary.
\end{proof}

One can easily obtain the main result of this subsection. 
\begin{PROP}\label{prop:apprcontrolled}
Let $G$ be a controlled graph with a partition $(D,K)$ and $x^*$ be a feasible fractional solution to  \sdhd\ such that $x^*_v< \frac{1}{20}$ for every $v\in V(G)$. There is a polynomial-time algorithm which, given such $G$ and $x^*$, finds a \dhm\ $X$ of size at most $O(\abs{x^*}^2)$.
\end{PROP} 
\begin{proof}
By Lemma~\ref{lem:newDHDsol}, we can obtain a feasible fractional solution $x'$ to \sdhd\ such that $x'_v=0$ for every $v\in K$, $x'_v<\frac{1}{10}$ for every $v\in D$ and $\abs{x'}\leq 2\abs{x^*}$. Such $x'$ meets the condition of Lemma~\ref{lem:usemulticut}, and we can obtain a \dhm\ $X$ of $G$ such that $\abs{X}=O(\abs{x'}^2)=O(\abs{x^*}^2)$ in polynomial time.
\end{proof}

\subsection{Proof of Theorem~\ref{thm:approx}}\label{subsec:finalstage}

We first present the proof of Proposition~\ref{prop:modulator}. 
Let $(G,k)$ be an instance of \sdhd\ such that $G$ does not contain any small DH obstructions. 
Let $x^*$ be an optimal fractional solution to \sdhd\ for $G$. We may assume that $\abs{x^*}\leq k$, otherwise we immediately report that $(G,k)$ is a \NO-instance. Let $\tilde{X}$ be the set of all vertices $v$ such that $x^*_v\geq \frac{1}{20}$. Observe that $\abs{\tilde{X}}\leq 20k$ since otherwise, $\abs{x^*}\geq \frac{1}{20}\cdot \abs{\tilde{X}}>k$, a contradiction. Also $x^*$ restricted to $V(G)\setminus \tilde{X}$  is a fractional feasible solution to \sdhd\ for $G-\tilde{X}$ such that $x^*_v<\frac{1}{20}$ for every $v$.

We compute a decomposition $V(G-\tilde{X})=D\uplus \bigcup_{i=1}^{\ell} K_i \uplus X$ as in Proposition~\ref{prop:decomposition}, or correctly report $(G,k)$ as a \NO-instance. Recall that $\ell=O(k\log n)$ and $\abs{X}=O(k^2\sqrt{\log{k}}\cdot \log n)$.  
Note that $V(G-(\tilde{X}\cup X))=D\uplus \bigcup_{i=1}^{\ell} K_i$. 
From $i=1$ up to $\ell$, we want to obtain a \dhm\ $S_i$ of $G_i$, where $G_{1}:=G[D\cup K_{1}]$ and for $i=2,\ldots , \ell$, $G_i$ is the subgraph of $G$ induced by $(V(G_{i-1})\setminus S_{i-1})\cup K_{i}$. Notice that  for every $i=1,\ldots , \ell-1$, if  $G_{i}-S_{i}$ is distance-hereditary, then $G_{i+1}$ is a controlled graph with a partition $(V(G_{i})\setminus S_{i}, K_{i+1})$. For $i=1$, clearly $G_{1}$ is a controlled graph. Hence,  we can inductively apply the algorithm of Proposition~\ref{prop:apprcontrolled} and obtain a \dhm\ $S_i$ of size at most $O(\abs{x^*}^2)$ of $G_i$. Especially, $G_{\ell}-S_{\ell}$ is distance-hereditary, implying that the set defined as 
\[S:=\tilde{X}\cup X\cup \bigcup_{i=1}^{\ell} S_i\]
is a \dhm\ of $G$. From $\abs{x^*}\leq k$, we have $\abs{S_i}=O(k^2)$. It follows that $\abs{S}=O(k^3\cdot \log n)$.

Proposition~\ref{prop:modulator} immediately yields the proof of Theorem~\ref{thm:approx}, which we summarize below.

\begin{proof}[Proof of Theorem~\ref{thm:approx}]
Let $\mathcal{P}$ be a maximal collection of vertex-disjoint copies of small DH obstructions in $G$. If $\mathcal{P}$ contains at least $k+1$ copies, then clearly $(G,k)$ is a \NO-instance to \sdhd. Otherwise, let $V(\mathcal{P})$ be the vertex set of the copies in $\mathcal{P}$ and notice that $\abs{V(\mathcal{P})}\leq 6k$. Notice that $G-V(\mathcal{P})$ contains no small DH obstruction, and thus we can apply the algorithm $\mathcal{A}$ of Propositioin~\ref{prop:modulator}. If $\mathcal{A}$ reports that $(G-V(\mathcal{P}),k)$ is a \NO-instance, then  clearly $(G,k)$ is a \NO-instance as well. Otherwise, let $S$ be a \dhm\ of $G-V(\mathcal{P})$ whose size is $O(k^3\cdot \log n)$. It remains to observe that $S\cup V(\mathcal{P})$ is a \dhm\ of $G$ whose size is $O(k^3\cdot \log n)$. 
\end{proof}

\section{Good Modulator}\label{sec:good}

In the previous section, we presented a polynomial-time algorithm, given an instance $(G,k)$ which outputs a \dhm\ of size $O(k^3\cdot \log n)$ whenever $(G,k)$ is a \YES-instance. In this section, we shall see how to obtain a good \dhm. In order to obtain a good \dhm, we need to find a small-sized hitting set that intersects every DH obstruction having exactly one vertex in $S$. This task is easy for a small DH obstruction, but not straightforward for the induced cycles of length at least 5. 
We first present a tool to efficiently handle the latter case.

\begin{PROP}\label{prop:sunflower}
Let $G$ be a graph without any small DH obstruction, $v$ be a vertex of $G$ such that $G-v$ is distance-hereditary and $k$ be a positive integer. In polynomial time, one can either 
\begin{enumerate}
\item find a set $X\subseteq V(G)\setminus \{v\}$ of size at most $O(k^2)$ such that $G-X$ contains no induced cycle of length at least 5 traversing $v$, or
\item correctly reports that any \dhm\ of size at most $k$ must contain $v$.
\end{enumerate}
\end{PROP}

\begin{proof}
Consider an instance $(G-v,\mathcal{T})$ of \textsc{Vertex Multicut} where \[\mathcal{T}:=\{(s,t): s,t\in N_G(v), \dist_{G-v}(s,t)\geq 3\}.\]

First, we claim that $X\subseteq V(G)\setminus \{v\}$ hits all induced cycles of $G$ of length at least 5 if and only if $X$ is a vertex multicut for $(G-v,\mathcal{T})$. Suppose $X$ is a vertex multicut for $(G-v,\mathcal{T})$ and  $H$ is an induced cycle of length at least 5 in $G-X$. Since $G-v$ is distance-hereditary,  $v$ is a vertex of $H$ and  $H-v$ is a path of length at least 3 between, say, $s$ and $t$. This means $\dist_{G-v}(s,t)\geq 3$ because $G-v$ is distance-hereditary and thus $(s,t)\in \mathcal{T}$, contradicting the assumption that $X$ hits all paths between every terminal pair in $\mathcal{T}$. Conversely, suppose $X\subseteq V(G)\setminus \{v\}$ hits all induced cycles of length at least 5 and there is an $(s,t)$-path $P$ in $(G-v)-X$ for some $(s,t)\in \mathcal{T}$. Note that $P$ is of length at least 3 and thus by Lemma~\ref{lem:createdhobs}, $G[\{v\}\cup V(P)]$ contains a DH obstruction $H$. By the assumption that $G$ does not contain a small DH obstruction, $H$ is an induced cycle of length at least 7 in $G-X$, a contradiction. 

Let $x^*$ be an optimal fractional solution to \textsc{Vertex Multicut}, which can be efficiently found using the ellipsoid method and an algorithm for the (weighted) shortest path problem as a separation oracle. If $\abs{x^*}\leq k$, then we can construct a multicut $X\subseteq V(G)\setminus \{v\}$ of size $O(\abs{x^*}^2)=O(k^2)$ using the approximation algorithm of Theorem~\ref{thm:fractionalmulticut}. By the previous claim, we know that $X$ hits all induced cycles of $G$ of length at least 5 (which must traverse $v$). 
If $\abs{x^*}>k$, then any integral solution for $(G-v,\mathcal{T})$ is larger than $k$. By the previous claim, any set $X\subseteq V(G)\setminus \{v\}$ hitting all induced cycles of $G$ of length at least 5 must be larger than $k$. It follows that any \dhm\ of size at most $k$ must contain $v$, completing the proof. 
\end{proof}

The following theorem states that a good \dhm\ of size $O(k^5\cdot \log n)$ can be constructed efficiently.

\begin{THM}\label{thm:goodmodulator}
There is a polynomial-time algorithm which, given a graph $G$ and a positive integer $k$, either 
\begin{enumerate}
\item returns an equivalent instance $(G',k')$ with a good \dhm\ $S'\subseteq V(G')$ of size $O(k^5\cdot \log n)$, or
\item correctly reports that $(G,k)$ is a \NO-instance to \dhd.
\end{enumerate}
\end{THM}
\begin{proof}

	We first apply the algorithm of Theorem~\ref{thm:approx}. If this algorithm reports that $(G,k)$ is a \NO-instance, then we are done. 
	Hence, we assume that a \dhm\ $S$ of $G$ containing at most $O(k^3\cdot \log n)$ vertices is returned. 
	Let $U:=\emptyset$, and for each $v\in S$, let $H_v:=G[(V(G)\setminus S)\cup \{v\}]$.
	
	For each $v\in S$, we do the following. 
	First, find either  $k+1$ small DH obstructions in $H_v$ whose pairwise intersection is $v$, 
	or a vertex set $T_v$ of $V(G)\setminus S$ such that $\abs{T_v}\le 5k$ and $H_v-T_v$ has no small DH obstructions.
	It can be done in polynomial time by going through all $5$-size subsets of $V(G)\setminus S$.
	In the former case, we add $v$ to $U$.
	Otherwise, we obtain a vertex set $T_v$ in the second statement.
	
	Assume we obtained the vertex set $T_v$.
	Since $H_v-T_v$ has no small DH obstructions, 
	every DH obstruction in $H_v$ is an induced cycle of length at least $7$.
	We apply the algorithm of Proposition~\ref{prop:sunflower} to $H_v-T_v$ and $v$, and it returns 
	either a vertex set $X_v\subseteq V(H_v-T_v)\setminus \{v\}$ of size $\mathcal{O}(k^2)$ such that $H_v-(T_v\cup X_v)$ has no DH obstructions, 
	or correctly reports that any DH-modulator of size at most $k$ must contain $v$.
	In the latter case, we add $v$ to $U$.
	Otherwise, we obtain such a vertex set $X_v$.
	This finishes the algorithm.
	
	We claim that $(G-U, k-\abs{U})$ is an instance equivalent to $(G,k)$ and $S\cup (\bigcup_{v\in S\setminus U} (T_v\cup X_v))$ is a good DH-modulator for $G-U$.
 	It is easy to see that if $(G-U, k-\abs{U})$ a \YES-instance, then $(G,k)$ is a \YES-instance.
	Suppose $G$ has a vertex set $T$ such that $\abs{T}\le k$ and $G-T$ is distance-hereditary.
	Let $v\in U$. Since $v\in U$, by the algorithm as above, 
	either there are $k+1$ pairwise small DH obstructions in $H_v$ whose pairwise intersection is $v$, 
	or any DH-modulator of size at most $k$ must contain $v$.
	Thus, $T$ contains $v$, and it implies that $U\subseteq T$.
	Therefore, $T\setminus U$ is a solution of $(G-U, k-\abs{U})$.

	Let $S':=S\cup (\bigcup_{v\in S\setminus U} (T_v\cup X_v))$.
	It remains to see that $S'$ is a good \dhm\ for $G-U$. Clearly $S'$ is a \dhm. 
	For contradiction, suppose that $F$ is a DH obstruction of $G-U$ and that $V(F)\cap S'=\{w\}$. 
	Since $S$ is a \dhm\ of $G$, $w$ cannot be a vertex of $\bigcup_{v\in S\setminus U} (T_v\cup X_v)$, and thus we have $w\in S$ and $F$ is an induced subgraph of $H_w-(T_w\cup X_w)$. 
	However, $H_w-(T_w\cup X_w)$ has no DH obstruction, contradiction.
	This completes the proof.
\end{proof}

We remark that given a graph $G$ and a good \dhm\, removing a vertex $v$ in $V(G)\setminus S$ does not create a new DH obstruction. Hence, $S$ remains a good \dhm\ in the graph $G-v$. 

\section{Twin Reduction Rule}\label{sec:twinred}

	In a distance-hereditary graph, there may be a large set of pairwise twins.
	We introduce a reduction rule that bounds the size of a set of pairwise twins in $G-S$ by $\mathcal{O}(k^2\abs{S}^3)$, 
	where $S$ is a DH-modulator.
	In the FPT algorithm obtained by Eiben, Ganian, and Kwon~\cite{EibenGK2016}, there is a similar rule which reduces the size of a twin set 
	outside of a DH-modulator, however, under the assumption that the given instance has no small DH obstructions. 
	For our kernelization algorithm, we cannot assume that the given instance has no small DH obstructions.
	Therefore, we need to analyze more carefully. 

	The underlying observation is that it suffices to keep up to $k+1$ vertices 
	that are pairwise twins with respect to each subset of $S$ of small size.  
	For a subset $S'\subseteq S$, two vertices $u$ and $v$ in $V(G)\setminus S$ are \emph{$S'$-twins} if $u$ and $v$ have the same neighbors in $S'$. 
	It is not difficult to get an upper bound $\mathcal{O}(k\abs{S}^5)$, by considering all subsets $S'$ of $S$ of size $\min\{\abs{S}, 5\}$ and marking up to $k+1$ $S'$-twins.
	To get a better bound, we proceed as follows.

\begin{RRULE}\label{rrule:exttwinreduction}
	Let $W$ be a set of pairwise twins in $G-S$, and let $m:=\min \{\abs{S}, 3\}$.
	\begin{enumerate}[(1)]
	\item Over all subsets $S'\subseteq S$ of size $m$, we mark up to $k+1$ pairwise $S'$-twins in $W$ that are unmarked yet. 
	\item When $\abs{S}\ge 4$, over all subsets $S'\subseteq S$ of size $4$, 
	if there is an unmarked vertex $v$ of $W$ such that $G[S'\cup \{v\}]$ is isomorphic to the house or the gem, then 
	we mark up to $k+1$ previously unmarked vertices in $W$ including $v$ that are pairwise $S'$-twins.
	\item If there is an unmarked vertex $v$ of $W$ after finishing the marking procedure, we remove $v$ from $G$.
	\end{enumerate}
\end{RRULE}

\begin{LEM}\label{lem:exttwinreduction}
Reduction Rule~\ref{rrule:exttwinreduction} is safe
\end{LEM}
\begin{proof}
	We recall that $m:=\min \{\abs{S}, 3\}$.
	Let $W$ be a set of pairwise twins in $G-S$, 
	and suppose there is an unmarked vertex in $W$.
	We show that $(G,k)$ is a \YES-instance if and only if $(G-v, k)$ is a \YES-instance.
	The forward direction is clear. Suppose that $G-v$ has a vertex set $T$ with $\abs{T}\le k$ such that $(G-v)-T$ is distance-hereditary.
	We claim that $G-T$ is also distance-hereditary. 
	For contradiction, suppose $G-T$ contains a DH obstruction $F$. Clearly, $v\in V(F)$.

	In case when $F-v$ is an induced path, let $w,z$ be the end vertices of the path, and choose a set $S'\subseteq S$ of size $m$ containing $\{w,z\}\cap S$ (possibly, an empty set).
	Since $v$ is an unmarked vertex in Reduction Rule~\ref{rrule:exttwinreduction}, 
	there are $v_1, \ldots, v_{k+1}\in W\setminus \{v\}$ where $v_1, \ldots, v_{k+1}, v$ are pairwise $S'$-twins.
	Note that $V(F)\cap \{v_1, \ldots, v_{k+1}\}=\emptyset$ since no other vertex in $F$ is adjacent to both $w$ and $z$.
	Thus, there exists $v'\in \{v_1, \ldots, v_{k+1}\}\setminus T$ such that $G[V(F)\setminus \{v\}\cup \{v'\}]$ is a DH obstruction in $(G-v)-T$, a contradiction.
	We may assume $F-v$ is not an induced path, and in particular $F$ is not an induced cycle.

	Suppose $\abs{V(F)\cap S}\le 3$. Then there exists a set $S'\subseteq S$ of size $m$ such that $V(F)\cap S\subseteq S'$.
	Since $v$ remains unmarked in the application of Reduction Rule~\ref{rrule:exttwinreduction}, 
	there are $v_1, \ldots, v_{k+1}\in W\setminus \{v\}$ where $v_1, \ldots, v_{k+1}, v$ are pairwise $S'$-twins.
	Note that $V(F)\cap \{v_1, \ldots, v_{k+1}\}=\emptyset$ since $F$ contains no twins.
	Thus, there exists $v'\in \{v_1, \ldots, v_{k+1}\}\setminus T$ such that $G[V(F)\setminus \{v\}\cup \{v'\}]$ is a DH obstruction in $(G-v)-T$, a contradiction.
	
	Now we assume $\abs{V(F)\cap S}\ge 4$. There are two possibilities.

\smallskip
\noindent\textbf{Case 1.} $F$ is isomorphic to either the house or the gem.

	Since $\abs{V(F)\cap S}\ge 4$, we have $V(F)\setminus \{v\}\subseteq S$.
	By Step (2) of Reduction Rule~\ref{rrule:exttwinreduction}, there are $v_1, \ldots, v_{k+1}\in W\setminus \{v\}$ where $v_1, \ldots, v_{k+1}, v$ are pairwise $(V(F)\cap S)$-twins.
	Thus, there exists $v'\in \{v_1, \ldots, v_{k+1}\}\setminus T$ such that $G[V(F)\setminus \{v\}\cup \{v'\}]$ is a DH obstruction in $(G-v)-T$, contradiction.

\noindent\textbf{Case 2.} $F$ is isomorphic to the domino.

	Since $F-v$ is not an induced path, $v$ is a vertex of degree $2$ in $F$.
	Let $v', w$ be the neighbors of $v$ having degree $2$ and degree $3$ in $F$, respectively. 
	Let $w'$ be the vertex of degree $3$ in $F$ other than $w$.
	Let $z$ be the vertex in $N_F(w)\setminus \{v, w'\}$ and let $z'$ be the vertex in $N_F(w')\setminus \{v', w\}$.
	
	Now, we take a subset $S'$ of $S$ of size $3$ containing $S\cap \{v',w,z'\}$ (possibly, an empty set).
	Since $v$ is an unmarked vertex in Reduction Rule~\ref{rrule:exttwinreduction}, 
	there are $v_1, \ldots, v_{k+1}\in W\setminus \{v\}$ where $v_1, \ldots, v_{k+1}, v$ are pairwise $S'$-twins.
	Note that $V(F)\cap \{v_1, \ldots, v_{k+1}\}=\emptyset$ since no other vertex in $F$ is adjacent to both $v', w$ and not adjacent to $z'$.
	Thus, there exists a vertex $v''\in \{v_1, \ldots, v_{k+1}\}\setminus T$.
	
	If $v''$ is adjacent to $z$, then $v'w'wz$ is an induced path of length 3 and $v''$ is adjacent to its end vertices. It follows that 
	$G[\{v'', v', w', w, z\}]$ contains a DH obstruction by Lemma~\ref{lem:createdhobs}, a contradiction to the assumption that $(G-v)-T$ is distance-hereditary.
	Suppose that $v''$ is not adjacent to $z$.
	If $v''$ is adjacent to $w'$, then $G[\{v'', w,w',z,z'\}]$ is isomorphic to the house, 
	and otherwise, $G[\{v'', v', w, w', z, z'\}]$ is isomorphic to the domino.
	This contradicts the assumption that $(G-v)-T$ is distance-hereditary.

\smallskip

	Therefore, $G-T$ is distance-hereditary, which completes the proof.
\end{proof}

	We can apply Reduction Rule~\ref{rrule:exttwinreduction} exhaustively in polynomial time by considering all twin sets of $G-S$, 
	and for each twin set going through all subsets $S'$ of $S$ of size $\min (\abs{S}, 4)$. 
	We observe that the size of a twin set $W$ after applying Reduction Rule~\ref{rrule:exttwinreduction} is bounded by a polynomial function in $k$ and $\abs{S}$.
	
\begin{LEM}\label{lem:twinsize}
	Let $S$ be a  DH-modulator of $G$, and $W$ be a set of pairwise twins in $G-S$. If $(G,k)$ is irreducible with respect to Reduction Rule~\ref{rrule:exttwinreduction}, 
	then we have $\abs{W}=\mathcal{O}(k^2\abs{S}^3)$ or $(G,k)$ is a \NO-instance.
\end{LEM}
\begin{proof}

	Suppose that $(G,k)$ is a \YES-instance. It suffices to prove that $\abs{W}$ satisfies the claimed bound.
	First assume that $\abs{S}\ge 4$.
	There are at most $\abs{S}\choose 3$ different choices of $S'$ in Step (1), and for each $S'$ we mark up to $2^3(k+1)$ vertices of $W$. 
	
	Consider an auxiliary hypergraph $\mathcal{H}$ on the vertex set $S$. A size-4 subset $S'$ of $S$ forms a hyperedge of $\mathcal{H}$ if and only if $S'$ is used to mark some vertex of $W$ in Step (2). Observe that if there exist $k+1$ vertex-disjoint hyperedges of $\mathcal{H}$, then there are $k+1$ vertex-disjoint copies of a house or a gem. As we assume that $(G,k)$ is a \YES-instance, a maximum packing of hyperedges has size at most $k$. Let $C$ be the vertices of $S$ that are contained in a maximal packing of hyperedges  and notice that $\abs{C}\leq 4k$. Since any hyperedge $e$ of $\mathcal{H}$ intersects with $C$, the following count on the maximum number of hyperedges is derived: 
	\[{4k \choose 4}+{4k \choose 3}{\abs{S}\choose 1}+{4k \choose 2} {\abs{S}\choose 2}+{4k \choose 1} {\abs{S}\choose 3}.\]
For each hyperedge $S'$ of $\mathcal{H}$, at most $2^4$ different sets of pairwise $S'$-twins, from which we mark up to $k+1$ vertices. It follows that
\[W \leq 2^3(k+1){\abs{S}\choose 3}+ 2^4(k+1)\left({4k \choose 4}+{4k \choose 3}{\abs{S}\choose 1}+{4k \choose 2} {\abs{S}\choose 2}+{4k \choose 1} {\abs{S}\choose 3}\right).\]

Assume that $\abs{S}\le 3$. Then we mark up to $2^{\abs{S}} (k+1)\le 8(k+1)$ vertices of $W$. Thus, if $\abs{W}> 8(k+1)$, then the set $W$ can be reduced further by removing a vertex.
The claimed bound follows.
\end{proof}

\section{Bounding the number of non-trivial connected components of $G-S$}\label{sec:countcc}

	We provide a reduction rule that that bounds the number of connected components of $G-S$ each having at least $2$ vertices, 
	when $S$ is a good DH-modulator.

	Let $(G,k)$ be an instance, and let $S$ be a good DH-modulator.
	For each pair of $v\in S$ and a connected component $C$ of $G-S$, let $N(v,C):=N_G(v)\cap V(C)$. 
	Note that for a non-trivial connected component $C$ of $G-S$, $(V(C), V(G)\setminus V(C))$ is {\sl not} a split 
	if and only if there exist $v,w \in S$ such that $N(v,C)\neq \emptyset$, $N(w,C)\neq \emptyset$ and $N(v,C)\neq N(w,C)$. 
	We say that a pair $(v,w)$ of vertices in $S$ is a \emph{witnessing pair} for a connected component $C$ of $G-S$
	if $N(v,C)\neq \emptyset$, $N(w,C)\neq \emptyset$ and $N(v,C)\neq N(w,C)$.
The following lemma is essential.

	\begin{LEM}\label{lem:twoccobs}
	Let $(G,k)$ be an instance and $S$ be a good DH-modulator.
	If $C_1$, $C_2$ are two connected components of $G-S$ and $v,w \in S$ such that $(v,w)$ is a witnessing pair for both $C_1$ and $C_2$,
	then $G[\{v,w\}\cup V(C_1)\cup V(C_2)]$ contains a DH obstruction.
	\end{LEM}
	\begin{proof}
	Since $N(v, C_1)\neq \emptyset$, $N(w, C_1)\neq \emptyset$, and $N(v, C_1)\neq N(w, C_1)$, 
	there exists a vertex $z\in V(C_1)$ where $z$ is adjacent to only one of $v$ and $w$.
	Without loss of generality, we may assume that $z$ is adjacent to $v$.
	We choose a neighbor $z'$ of $v$ in $C_2$.
	Clearly, there is a path from $z$ to $z'$ in $G[\{v,w\}\cup V(C_1)\cup V(C_2)]-v$ because $w$ has a neighbor on each of $C_1$ and $C_2$. 
	Let $P$ be a shortest path from $z$ to $z'$ in $G[\{v,w\}\cup V(C_1)\cup V(C_2)]-v$.
	Note that $P$ has length at least $3$, as $w$ and $z$ are not adjacent.
	Thus, by Lemma~\ref{lem:createdhobs},  $G[\{v,w\}\cup V(C_1)\cup V(C_2)]$ contains a DH obstruction.
	\end{proof}

Lemma~\ref{lem:twoccobs} observes that if a pair of vertices in $S$ witnesses at least $k+2$ non-trivial connected components in $G-S$, at least one of the pair must be contained in any size-$k$ \dhm. Furthermore, keeping exactly $k+2$ non-trivial connected components would suffice to impose this restriction. This suggests the following reduction rule.

\begin{RRULE}\label{rrule:boundingcc} 
For each pair of vertices $v$ and $w$ in $S$, we mark up to $k+2$ non-trivial (previously unmarked) connected components $C$ of $G-S$ such that
$(v, w)$ is a witnessing pair for $C$. If there is an unmarked non-trivial connected component $C$ after the marking procedure, then we remove all edges in $C$. 

\end{RRULE}

The following is useful to see the safeness of our reduction rule. 
\begin{LEM}\label{lem:splitobs}
Let $G$ be a graph, and let $(A,B)$ be a split of $G$. Then for every DH obstruction $H$ of $G$, 
either $\abs{V(H)\cap A}\le 1$ or $\abs{V(H)\cap B}\le 1$.
\end{LEM}
\begin{proof}
Suppose for contradiction that $\abs{V(H)\cap A}\ge 2$ and $\abs{V(H)\cap B}\ge 2$.
Since $H$ is connected, $(V(H)\cap A, V(H)\cap B)$ is a split of $H$.
This is contradiction because every DH obstruction does not have a split.
\end{proof}

\begin{LEM}\label{lem:boundingcc}
Reduction Rule~\ref{rrule:boundingcc} is safe. Moreover, $S$ remains a good \dhm\ after applying the reduction rule.
\end{LEM}
\begin{proof}

	Suppose there is an unmarked non-trivial connected component $C$ of $G-S$ after the marking procedure.
	Let $G'$ be the graph obtained by removing all edges in  $C$.
	We claim that $(G, k)$ is a \YES-instance if and only if $(G', k)$ is a \YES-instance.
	
	We first prove the converse direction.
	Suppose $T$ is a vertex set of $G'$ such that $\abs{T}\le k$ and $G'-T$ is distance-hereditary. 
	We claim that $(V(C)\setminus T, V(G)\setminus V(C)\setminus T)$ is a split of $G-T$.
	Suppose $(V(C)\setminus T, V(G)\setminus V(C)\setminus T)$ is not a split of $G-T$.
	Thus, there are $v,w\in S\setminus T$, where $N_G(v)\cap (V(C)\setminus T)\neq \emptyset$, $N_G(w)\cap (V(C)\setminus T)\neq \emptyset$, 
	and $N_G(v)\cap (V(C)\setminus T)\neq N_G(w)\cap (V(C)\setminus T)$.
	It further implies that $(v,w)$ is a witnessing pair for $C$ in $G$.
	Since $C$ is unmarked, there are $k+2$ non-trivial connected components $C_1, \ldots, C_{k+2}$ of $G-S$ other than $C$ 
	where $(v,w)$ is a witness pair for each $C_i$ in $G$.
	From $\abs{T}\leq k$, there are two components  $C_i, C_j\in \{C_1, \ldots, C_{k+2}\}$ that do not contain a vertex of $T$. 
	By Lemma~\ref{lem:twoccobs}, the graph $(G'-T)[\{v,w\}\cup V(C_1)\cup V(C_2)]$ contains a DH obstruction.
	It contradicts to the assumption that $G'-T$ has no DH obstructions.
	Thus, $(V(C)\setminus T, V(G)\setminus V(C)\setminus T)$ is a split of $G-T$ as claimed. 

	Suppose for contradiction that $G-T$ contains a DH obstruction $H$. 
	$(V(C)\setminus T, V(G)\setminus V(C)\setminus T)$ is a split of $G-T$, 
	we have either $\abs{V(H)\cap V(C)}\leq 1$ or $\abs{V(H)\setminus V(C)}\le 1$ by Lemma~\ref{lem:splitobs}. 
	As $S$ is a good DH-modulator, the only possibility is $\abs{V(H)\cap V(C)}\leq 1$. 
	This means that $H$ is also an induced subgraph of $G'-T$, a contradiction. Therefore, $T$ is also a solution to $(G,k)$.

	For the forward direction, suppose that $T$ is a vertex set of $G$ such that $\abs{T}\le k$ and $G-T$ is distance-hereditary. 
	By a similar argument as above, one can show that $(V(C)\setminus T, V(G')\setminus V(C)\setminus T)$ is a split in $G'-T$. 
	Hence, any DH obstruction $H$ in $G'-T$ contains at most one vertex of $C$, and thus $H$ is also an induced subgraph of $G-T$, 
	which contradicts to the assumption that $G-T$ is distance-hereditary.
	We conclude that $(G',k)$ is a \YES-instance if $(G,k)$ is a \YES-instance. 

	To see that $S$ remains a good \dhm\ after applying the reduction rule,   suppose that $S$ is not a good \dhm\ after the application of Reduction Rule~\ref{rrule:boundingcc}. 
	It is easy to see that $S$ is again a \dhm\ of $G'$. Hence, we may assume there is a DH osbtruction $F$ in $G'$ 
	such that $\abs{V(F)\cap S}=1$. Such $F$ must contain at least two vertices $u,v$ such that $uv\in E(G)\setminus E(G')$. 
	Notice that the reduction rule ensures that such $u,v$ are adjacent only with the vertices of $S$. 
	Due to the assumption $\abs{V(F)\cap S}=1$, $u$ and $v$ are pendant vertices in $F$, which is not possible for a DH obstruction $F$. 
	Therefore, $S$ is a good \dhm\ of $G'$. This completes the proof.
\end{proof}

Now we analyze the bound on the number of non-trivial connected components in $G-S$ after applying Reduction Rule~\ref{rrule:boundingcc}. It is not difficult to see that there are at most $(k+2)\abs{S}^2$ non-trivial connected components in $G-S$ after the reduction. We obtain the following bound, which is better than a naive bound whenever $\abs{S}$ is a polynomial in $k$ of degree at least two.

\begin{PROP}\label{prop:ccnumber}
Let $(G,k)$ be an instance and $S$ be a good DH-modulator. If $(G,k)$ is irreducible with respect to Reduction Rule~\ref{rrule:boundingcc}, then either the number of non-trivial connected components is at most $O(k^2\cdot \abs{S})$ or it is a \NO-instance.
\end{PROP}

The basic idea is that if there is a collection of $k+1$ pairwise disjoint pairs in $S$, each witnessing at least two non-trivial connected components of $G-S$, there exist $k+1$ disjoint copies of DH obstructions due to Lemma~\ref{lem:twoccobs}. Therefore, after removing the vertices in this {\sl matching} structure, any pair of vertices in the remaining part of $S$ must witness at most one connected component. Then, we use  Erd\H{o}s-P\'osa property for cycles and a generalization of Lemma~\ref{lem:twoccobs} to argue that the remaining part of $S$ is {\sl sparse}, in the sense that the number of connected components witnessed by this part is linear in $\abs{S}$.

\begin{THM}[Erd\H{o}s and P\'osa~\cite{ErdosP1965}]\label{thm:EPcycle}
There exists a constant $r$  such that, given an arbitrary graph $G$ and a positive integer $k$, either finds $k+1$ vertex-disjoint cycles or a vertex set $X
\subseteq V(G)$ with $\abs{X}\leq r\cdot k\log k$ hitting all cycles of $G$. 
\end{THM}

The following is a generalization of Lemma~\ref{lem:twoccobs}. 

	\begin{LEM}\label{lem:cycleccobs}
	Let $(G,k)$ be an instance and $S$ be a good DH-modulator.
	Suppose that $v_1, v_2, \ldots, v_m$ is a circular ordering of $m(\geq 2)$ vertices in $S$ and $\mathcal{C}$ be a set of non-trivial connected components of $G-S$ such that 
	over all $i$, the pairs $(v_i, v_{i+1})$ witness mutually distinct components of $\mathcal{C}$, where $v_{m+1}=v_1$. 
	Then $G[\{v_1, v_2, \ldots, v_m\}\cup \bigcup_{C\in \mathcal{C}}V(C)]$ contains a DH obstruction.
	\end{LEM}
	\begin{proof}
	We prove by induction on $m$. Let $C_1,\ldots , C_m$ be mutually distinct components of $\mathcal{C}$ witnessed by $(v_i,v_{i+1})$. Let $H:=G[\{v_1, v_2, \ldots, v_m\}\cup \bigcup_{C\in \mathcal{C}}V(C)]$.
	The statement is equivalent to Lemma~\ref{lem:twoccobs}  when $m=2$, hence we assume $m\geq 3$. 

	We claim that for every $i$, the neighbors of $C_i$ among $\{v_1,\ldots , v_m\}$ are exactly $\{v_i,v_{i+1}\}$ or there is
	a DH obstruction in $H$. 
	Suppose that there exist $i$ and $v_j\notin \{v_i,v_{i+1}\}$ such that $v_j$ has a neighbor in $C_i$. 
	If $N(v_j, C_i)=N(v_i, C_i)$, then $N(v_j, C_i)\neq N(v_{i+1}, C_i)$, and thus  
	$(v_j, v_{i+1})$ is a witnessing pair for $C_i$. Observe that $v_{i+1}, v_{i+2}, \ldots , v_j$ is a circular ordering (possibly $j=i+2$) meeting the condition of the statement, with $C_{i+1},\ldots , C_i$ as the connected components witnessed by the consecutive pairs $(v_{i+1},v_{i+2}),\ldots , (v_j,v_{i+1})$. 
	If $N(v_j, C_i)\neq N(v_i, C_i)$, the circular ordering $v_i,v_j,\ldots , v_{i-1}$, possibly with $j=i-1$, meets the condition with $C_i,C_j,\ldots , C_{i-1}$ witnessed by the pairs $(v_i,v_j), (v_j,v_{j+1}),\ldots , (v_{i-1},v_i)$. In both cases, we conclude that $H$ contains a DH obstruction by induction hypothesis.
	
	Therefore, we may assume that for every $i$, the neighbors of $C_i$ among $\{v_1,\ldots , v_m\}$ are exactly $\{v_i,v_{i+1}\}$. 
	Since $N(v_1, C_1)\neq \emptyset$, $N(v_2, C_1)\neq \emptyset$, and $N(v_1, C_1)\neq N(v_2, C_1)$, 
	there exists a vertex $z\in V(C_1)$ where $z$ is adjacent to only one of $v_1$ and $v_2$.
	By symmetry, we may assume that $z$ is adjacent to $v_1$.
	We choose a neighbor $z'$ of $v_1$ in $C_m$.
	Note that $H-v_1$ is connected. Let $P$ be a shortest path from $z$ to $z'$ in $H-v_1$. 
	Since $\{v_2,\ldots , v_m\}$ separates $z$ and $z'$ in $H-v_1$ and the only neighbor of $C_1$ in $\{v_2,\ldots , v_m\}$ is $v_2$ by the above claim, the path $P$ traverse $v_2$. From the fact that $z$ is not adjacent with $v_2$, it follows that $P$ has length at least 3. By Lemma~\ref{lem:createdhobs}, $G[V(P)\cup \{v_1\}]$ contains a DH obstruction. This completes the proof. 
	\end{proof}

We are ready to prove Proposition~\ref{prop:ccnumber}

\begin{proof}[Proof of Proposition~\ref{prop:ccnumber}]
If $(G,k)$ is a \NO-instance, there is nothing to prove. We assume that $(G,k)$ is a \YES-instance. 

Let us define an auxiliary multigraph $F$ on the vertex set $S$ such that for every pair $v,w\in S$, the multiplicity of the edge $vw$ equals the number of non-trivial connected components that are marked by the witness of $(v,w)$ in Reduction Rule~\ref{rrule:boundingcc}. Recall that each non-trivial connected component of $G-S$ is marked at most once, which implies that the edge set of $F$ can be bijectively mapped to the non-trivial connected components of $G-S$ (after removing unmarked components). Therefore, it suffices to obtain a bound on the number of edges in $F$ with the edge multiplicity taken into account.

Construct a maximal packing of 2-cycles in $F$ and let $S_1\subseteq S$ be the vertices contained in the packing. By Lemma~\ref{lem:twoccobs}, a packing of size $k+1$ implies the existence of $k+1$ vertex-disjoint DH obstructions. Therefore, $\abs{S_1}\leq 2k$. 

Again, due to the assumption that $(G,k)$ is a \YES-instance, the subgraph $F-S_1$ does not have $k+1$ vertex-disjoint cycles: otherwise, $G$ contains $k+1$ vertex-disjoint DH obstructions by Lemma~\ref{lem:cycleccobs}. Theorem~\ref{thm:EPcycle} implies that there exists a vertex set $S_2\subseteq V(F)\setminus S_1$ hitting all cycles of $F-S_1$ with $\abs{S_2}\leq r\cdot k
\log k$  for some constant $r$. 

Now, the number of  edges in $F$  is at most 
\begin{align*}
&\abs{S_1}\cdot \abs{S}(k+2)+ \abs{S_2}\cdot \abs{S\setminus S_1} + (\abs{S\setminus S_1\setminus S_2})\\
&=	2k(k+2)\abs{S} + r\cdot k\log k\cdot \abs{S} + \abs{S}\leq (7+r)k^2\cdot \abs{S},
\end{align*}
which establishes the claimed bound.
\end{proof}

\section{Bounding the size of non-trivial connected components of $G-S$}\label{sec:boundingnontrivialcc}

	In this section, we introduce several reduction rules that reduce the number of bags in the canonical split decomposition of a non-trivial connected component of $G-S$, 
	where $S$ is a good DH-modulator of $G$. Together with Reduction Rule~\ref{rrule:exttwinreduction} which shrinks the twin set in each bag, 
	the result of this section bounds the number of vertices of each non-trivial connected component of $G-S$.
	In particular, we will prove that if the canonical split decomposition of a non-trivial connected component of $G-S$ contains 
	more than $3\abs{S}(20k+54)$ bags, then we can apply some reduction rule.

	Let $D$ be the canonical split decomposition of a connected component $H$ of $G-S$.
	Since $S$ is a good DH-modulator, 
	for each vertex $v$ of $S$, $G[V(H)\cup \{v\}]$ is distance-hereditary.
	In particular, if $G[V(H)\cup \{v\}]$ is connected,
	then we can extend $D$ into a canonical split decomposition of $G[V(H)\cup \{v\}]$ using the result by Gioan and Paul~\cite{GioanP2012}.
	By Theorem~\ref{thm:GioanP2012}, in time $\mathcal{O}(\abs{V(G)})$ 
	we can uniquely decide either a bag or a marked edge of $D$ that will {\sl accommodate} $v$ so as to extend the split decomposition $D$ of $H$ into the split decomposition of $G[V(H)\cup \{v\}]$. See Theorem~\ref{thm:GioanP2012} for details on how the unique bag (called the partially accessible bag) or marked edge is characterized.
	The bag or marked edge that accommodations $v$ shall be colored for all $v\in S$, and also we apply a short-range propagation rule on the colored parts. We shall apply reduction rules to shrink uncolored parts.

	Let $f$ be the injective function from $S\cap N_G(V(H))$ to the union of the set of all marked edges and the set of all bags such that
	$f(v)$ is the the partially accessible bag or the marked edge indicated by Theorem~\ref{thm:GioanP2012}.
	We can compute this function $f$ in time $\mathcal{O}(\abs{S} \cdot \abs{V(G)})$.
	A bag or a marked edge in $D$ is \emph{$S$-affected} if it is $f(v)$ for some $v\in S$, and \emph{$S$-unaffected} otherwise.

	We may assume that $D$ consists of more than one bag.
	For two adjacent bags $B_1$ and $B_2$, we denote by $e(B_1, B_2)$ the marked edge linking $B_1$ and $B_2$.

\subsection{Reduction Rules to bound the number of leaf bags}\label{subsec:leafbags}

We first provide simpler rules.
\begin{RRULE}\label{rrule:removeleaf1}
If $v$ is a vertex of degree at most 1 in $G$, then remove $v$.
\end{RRULE}
The safeness of Reduction Rule~\ref{rrule:removeleaf1} is clear.

\begin{figure}
\includegraphics[scale=0.35]{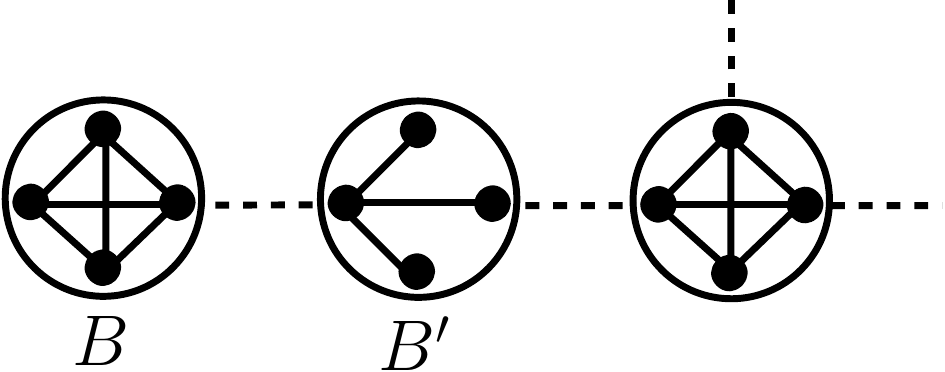} \qquad \qquad
\includegraphics[scale=0.35]{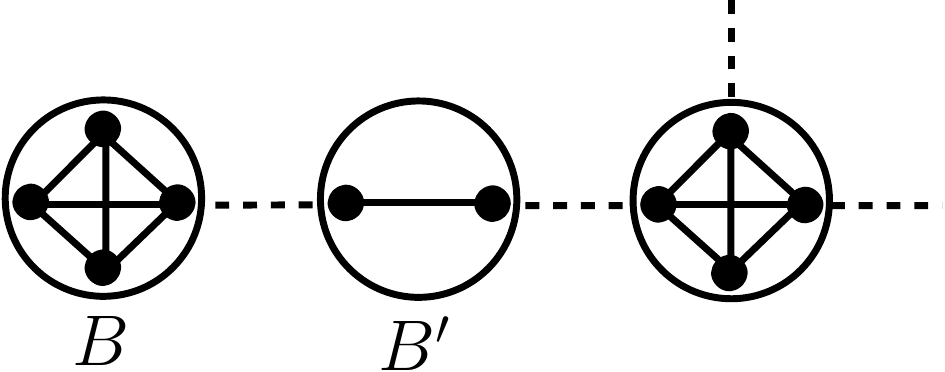}
\caption{An illustration of Reduction Rule~\ref{rrule:removeleaf2}. }
\label{fig:removeleaf2}
\end{figure}

\begin{RRULE}\label{rrule:removeleaf2}
Let $B$ be a leaf bag of $D$ and let $B'$ be its neighbor bag such that 
\begin{itemize}
\item $B$, $B'$, and $e(B,B')$ are $S$-unaffected, 
\item $D-V(B')$ has exactly two connected components, and
\item $B'$ is a star bag whose center is adjacent to $B$.
\end{itemize}
Then remove the unmarked vertices in $B'$.
\end{RRULE}
Reduction Rule~\ref{fig:removeleaf2} is illustrated in Figure~\ref{fig:removeleaf2}.

\begin{LEM}\label{lem:saferemoveleaf}
Reduction Rule~\ref{rrule:removeleaf2} is safe.
\end{LEM}
\begin{proof}
	We claim that there is no DH obstruction containing an unmarked vertex in $B'$. 
	First observe that $B$ is not a star bag whose leaf is adjacent to $B'$ because $D$ is a canonical split decomposition.
	Furthermore, vertices in $B$ are pairwise twins in $G$, as $B$ is $S$-unaffected.
	If $G$ has a DH obstruction $H$ containing an unmarked vertex of $B'$, then $H$ contains at least $2$ vertices in $B$ as $B$, $B'$, and $e(B, B')$ are $S$-unaffected.
	But this is not possible since every DH obstruction does not contain twins.
	Therefore, we can safely remove the vertices in $B'$.
\end{proof}

\begin{RRULE}\label{rrule:fliptwins}
	Let $A\subseteq V(H)$  be a set of vertices that are pairwise twins in $G$.
	Flip the adjacency relation among the vertices of $A$ if the resulting graph has strictly smaller number of bags in $D$.
\end{RRULE}

The safety of Reduction Rule~\ref{rrule:fliptwins} is already observed in~\cite{EibenGK2016}, which is an immediate consequence of the fact that every DH obstruction contains at most one vertex among pairwise twins. 

\begin{LEM}\label{lem:polytime}
One can apply Reduction Rules~\ref{rrule:removeleaf1}, \ref{rrule:removeleaf2}, and~\ref{rrule:fliptwins} exhaustively in polynomial time. Furthermore, $S$ is a good \dhm\ in the resulting graph $G'$.
\end{LEM}
\begin{proof}
	The statement trivially holds for Reduction Rules~\ref{rrule:removeleaf1} and \ref{rrule:removeleaf2}. 
	Let $A\subseteq V(H)$  be a set of vertices that are pairwise twins in $G$. 
	Then $A$ must be contained in a bag of $D$. 
	Notice that a twin set in $G-S$ is a twin set in $G$ if and only if it is  either complete or anti-complete to each vertex of $S$, which can be tested in polynomial time. 
	Therefore, by skimming through each bag and a twin set contained in it (and modifying the canonical split decomposition accordingly),
	we can detect a set $A$ to apply Reduction Rule~\ref{rrule:fliptwins} or correctly decide that no such $A$ exists. 
	Each application of Reduction Rules~\ref{rrule:removeleaf1}, \ref{rrule:removeleaf2}, and~\ref{rrule:fliptwins} reduces either the number of vertices in $G$ or the number of bags, 
	and the statement follows.

To see the second statement holds, we only need to check that $S$ is a good \dhm\ in $G'$ obtained by applying Reduction Rule~\ref{rrule:fliptwins} once. This is an immediate consequence of the fact $A$ is a twin set in $G'$ and no DH obstruction can contain two twin vertices.
\end{proof}

\begin{LEM}\label{lem:centermarked}
	Let $(G,k)$ be an instance reduced under Reduction Rules~\ref{rrule:removeleaf1}, \ref{rrule:removeleaf2},  and~\ref{rrule:fliptwins}. 
	If a leaf bag $B$ of $D$ is $S$-unaffected, then 
\begin{enumerate}[(i)]
\item either $B$ is a complete bag or it is a star bag whose center is marked, and
\item the unique bag $B'$ adjacent with $B$ is a star bag whose center is adjacent with $B$.
\end{enumerate}
\end{LEM}
\begin{proof}
	Suppose $B$ is a star bag whose center is unmarked. Since $B$ is $S$-unaffected, there are no edges between $S$ and the set of leaves of $B$.
	Thus, each leaf of $B$ is a leaf of $G$, contradicting that $(G,k)$ is reduced under Reduction Rule~\ref{rrule:removeleaf1}. 
	The first statement follows.

	To prove (ii), first suppose $B'$ is a complete bag. 
	Then $B$ cannot be a complete bag since there cannot be  two complete bags adjacent in a canonical split decomposition. 
	Hence by (i), $B$ is a star bag whose center is adjacent with $B'$. Since $B$ is $S$-unaffected, 
	each vertex of $S$ is either complete to the set of unmarked vertices of $B$, or anti-complete to it. 
	So, the set of unmarked vertices of $B$ is a twin set of $G$ and applying Reduction Rule~\ref{rrule:fliptwins} 
	will recompose the marked edge $e(B,B')$, strictly reducing the number of bags. 
	This contradicts that $(G,k)$ is reduced under Reduction Rule~\ref{rrule:fliptwins}.

	Therefore, $B'$ is a star bag. 
	First assume that a leaf of $B'$ is adjacent to $B$. 
	In that case, $B$ cannot be a star bag whose center is adjacent to $B'$, 
	because $D$ is a canonical split decomposition. Thus, $B$ is a complete bag. 
	Since $B$ is $S$-unaffected, the set of unmarked vertices of $B$ is a twin set of $G$, 
 	and applying Reduction Rule~\ref{rrule:fliptwins} will lead to recomposing the marked edge $e(B,B')$, 
	strictly reducing the number of bags. 
	Thus we conclude that the center of $B'$ is adjacent to $B$.
\end{proof}

	Throughout the rest of the section, we assume that $(G,k)$ is reduced under Reduction Rules~\ref{rrule:removeleaf1}, \ref{rrule:removeleaf2},  and~\ref{rrule:fliptwins}. 
	We call that a bag $B$ is a \emph{branch bag} if $D-V(B)$ contains at least $3$ connected components having at least two bags.
	We color the bags of $D$ with red and blue in the following way. 
	\begin{enumerate}
	\item If a bag $B$ is $S$-affected or incident with an $S$-affected edge, we color $B$ with red.
	\item If a bag $B$ is adjacent to an $S$-affected leaf bag,  we color $B$ with red.
	\item If $B$ is a branch bag, then we color $B$ with red. 
	\item All other bags are colored with blue.
	\end{enumerate}
	
	Let $\mathcal{R}$ be the set of all red bags, and let $\mathcal{Q}$ be the set of all blue leaf bags $B$ whose unique  neighbor bag is red. 
	We prove the followings.
	
	\begin{LEM}\label{lem:bluebags}
	For each bag $B$, there is at most one blue leaf bag adjacent to $B$.
	\end{LEM}
	\begin{proof}
	Suppose for contradiction that there are a bag $B$ and two blue leaf bags $B_1$ and $B_2$ adjacent to $B$.
	Since each $B_i$ is a blue leaf bag, $B_i$ is $S$-unaffected.
	Therefore, by Lemma~\ref{lem:centermarked}, $B$ is a star bag whose center is adjacent to $B_i$ for both $i=1,2$, which is impossible.
	\end{proof}
	
	\begin{LEM}\label{lem:bluepath2}
	Let $(G,k)$ be an instance reduced under Reduction Rules~\ref{rrule:removeleaf1}, \ref{rrule:removeleaf2},  and~\ref{rrule:fliptwins}. 
	For any connected components $D'$ of $D-\bigcup_{B\in \mathcal{R} \cup \mathcal{Q}}V(B)$, $D'$ is adjacent with exactly two red bags in $D$.
	\end{LEM}
	\begin{proof}
	Suppose that $D'$ is adjacent with exactly one red bag, say $L$, in $D$.
	Notice that $D'$ consists of at least two bags since otherwise, the single bag in $D'$ is in $\mathcal{Q}$.

	Let $B$ be a leaf bag of $D'$ that is farthest from $L$ in $D$ and $B'$ be the unique bag adjacent with $B$ in $D'$. Note that $B'\neq L$.
	Since $B$ is $S$-unaffected, Lemma~\ref{lem:centermarked} implies that  either $B$ is a complete bag or it is a star bag whose center is marked, 
	and $B'$ is a star bag whose center is adjacent with $B$. 
	
	We observe that $D-V(B')$ has exactly two connected components; the one consisting of $B$ and another component 
	adjacent with $L$ in $D$. Indeed, as we chose $B$ as a farthest leaf bag from $L$ in $D$, an additional component adjacent with $B'$ must consist of a single (leaf) bag. 
	However, this is impossible since $B$ is the only leaf bag of $D$ adjacent with $B'$ by Lemma~\ref{lem:bluebags}. 
	Now, $B$ and $B'$ satisfy the condition of Reduction Rule~\ref{rrule:removeleaf2}, contradicting the assumption that $(G,k)$ is reduced under this rule.
	
	Suppose that $D'$ is adjacent with at least three red bags, say $R_1,R_2$ and $R_3$, in $D$.
	Since $D'$ consists solely of blue bags, it does not contain a branch bag. Hence, at least one of $R_i$'s, say $R_1$, is a leaf bag of $D$. 
	However, $R_1$ could have been colored red only in (1), that is, $R_1$ is incident with $S$-affected edge or it is $S$-affected itself.
	In both cases, the bag of $D'$ adjacent with $R_1$ must be colored either by (1) or (2) of the coloring procedure, a contradiction.
	This completes the proof.
	\end{proof}
	
	\begin{LEM}\label{lem:redbags}
	Let $(G,k)$ be an instance reduced under Reduction Rules~\ref{rrule:removeleaf1}, \ref{rrule:removeleaf2},  and~\ref{rrule:fliptwins}. 
	Then, the number of red bags in $D$ is at most $3\abs{S}$.
	\end{LEM}
	\begin{proof}
	An $S$-affected edge causes at most two bags to be colored red in (1). An $S$-affected bag $B$ causes at most two bags to be colored 
	red in (2), that is, the bag itself and possibly an adjacent bag if $B$ is a leaf bag. Hence, the number of bags colored in (1)-(2), is at most $2\abs{S}$. 
	
	It remains to prove that the number of branch bags is at most $\abs{S}$.
	Let $\mathcal{B}$ be the set of branch bags. 
	We create a graph $F$ on the vertex set $\mathcal{B}$ such that two bags $B_1$ and $B_2$ in $\mathcal{B}$ are adjacent in $F$ 
	if and only if there is a path of bags in $D$ from $B_1$ to $B_2$ that does not contain any other branch bag. Since such a path of bags is unique between any $B_1$ and $B_2$ in $D$,
	$F$ is clearly a tree. 

	Let $L$ be a leaf of $F$. Since $D-V(L)$ contains at least $3$ connected components having at least two bags, 
	at least two of these components does not contain any other bag in $\mathcal{B}$. 
	Consider an arbitrary component $D'$ of $D-V(L)$ that contains at least two bags and does not contain a bag of $\mathcal{B}$.
	Lemma~\ref{lem:bluepath2} implies that $D'$ contains a red bag. We argue that $D'$  contains either an $S$-affected bag or an $S$-affected marked edge. 
	Suppose not, and notice that the only red bag in $D'$ is the one adjacent with $L$ in $D$. Let us call this bag $B$.
	By Lemma~\ref{lem:bluepath2}, the connected components of $D'-V(B)$ must be singletons, consisting of a leaf bag. 
	There are at most one such component by Lemma~\ref{lem:bluebags} consisting of a leaf bag $B'$. Then we can apply Reduction Rule~\ref{rrule:removeleaf2} 
	to $B$ and $B'$, a contradiction. We conclude that $D'$  contains either an $S$-affected bag or an $S$-affected marked edge.

	Now, every leaf $L$ of $F$ can be associated with two connected components of $D-V(L)$ containing either  an $S$-affected bag or an $S$-affected marked edge. 
	Notice that the two connected components, as a set of bags of $D$, associated to the set of all leaves in $F$ are pairwise disjoint. Therefore, the number of leaves of $F$ is at most $\frac{1}{2}\abs{S}$. 
	It follows that 
	$F$ contains at most $\abs{S}$ vertices, thus the number of red bags colored in (3) is at most $\abs{S}$. This completes the proof.
	\end{proof}

	\begin{LEM}\label{lem:longbluepath}
	The number of connected components of $D-\bigcup_{B\in \mathcal{R} \cup \mathcal{Q}}V(B)$ is at most $3\abs{S}$ and $\abs{\mathcal{R} \cup \mathcal{Q}}\leq 6\abs{S}$.
	\end{LEM}
	\begin{proof}
	Consider a tree $F$ on the vertex set $\mathcal{R}$ in which any two red bag is adjacent if and only if there is a path of bags in $D$ containing no red bag. 
	By Lemma~\ref{lem:bluepath2}, there is an injection from the set of connected components of $D-\bigcup_{B\in \mathcal{R} \cup \mathcal{Q}}V(B)$
	to the set of edges of $F$. Observe that $F$ has at most $\abs{\mathcal{R}}$ edges. By Lemma~\ref{lem:redbags}, we have $\abs{\mathcal{R}} \le 3\abs{S}$, which 
	establishes the first claimed bound. Lemma~\ref{lem:bluebags} implies that $\abs{\mathcal{Q}}\leq \abs{\mathcal{R}}$. The second bound $\abs{\mathcal{R}\cup \mathcal{Q}}\leq 6\abs{S}$ follows immediately.
	\end{proof}

\subsection{Reducing the length of a sequence of $S$-unaffected bags}\label{subsec:unaffectedbag}

	We assume that the instance $(G,k)$ is reduced under Reduction Rules~\ref{rrule:removeleaf1}, \ref{rrule:removeleaf2}, and \ref{rrule:fliptwins}.
	Let $D'$ be a connected component of $D-\bigcup_{B\in \mathcal{R}\cup \mathcal{Q}}V(B)$.
	We introduce three reduction rules that contribute to bound the number of bags in $D'$.
	Let us fix the sequence $B_1, \cdots ,B_m$ of bags of $D'$.
	For $1\le i<j<k\le m$, a bag $B_j$ is called a \emph{$(B_i,B_k)$-separator bag} if the center of $B_j$ is adjacent to neither $B_{j-1}$ nor $B_{j+1}$.
	For a $(B_1, B_m)$-separator bag $B$ we define 
	\begin{itemize}
	\item $\eta(B):=\{v\}$ if $v$ is the center of $B$ and there is no leaf bag adjacent to $B$, 
	\item $\eta(B):=X$ if there is a leaf bag $B'$ adjacent to $B$, and $X$ is the set of unmarked vertices in $B'$.
	\end{itemize}

	Below, we present a rule that gives a bound on the number of $(B_1, B_m)$-separator bags.
	We remark that for every bag $B_j$ in $D'$ with $1<j<m$ other than $(B_1, B_m)$-separator bags,
	there are no leaf bags having $B_j$ as a neighbor bag: if $B_j$ has a leaf bag as a neighbor, 
	such a leaf bag must be $S$-affected because of Lemma~\ref{lem:centermarked}. 
	Then $B_j$ is a red bag by the coloring procedure (3), contradicting that $B_j$ is a bag in $D'$.

\begin{figure}
\includegraphics[scale=0.35]{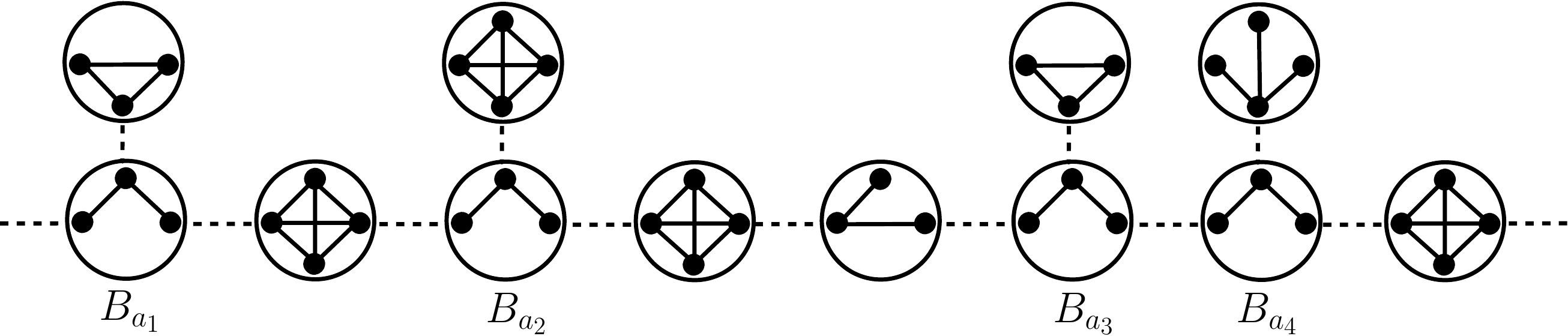} \vskip  0.3cm
\includegraphics[scale=0.35]{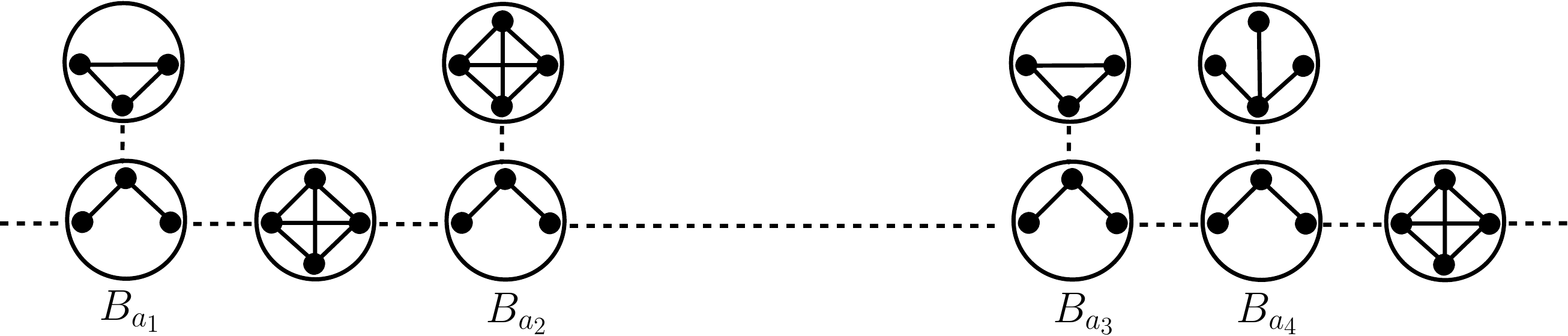}
\caption{An illustration of Reduction Rule~\ref{rrule:bypassing3}. }
\label{fig:bypassing3}
\end{figure}

\begin{RRULE}\label{rrule:bypassing3}
	Let $B_{a_1}, B_{a_2}, B_{a_3}, B_{a_4}$  be $(B_1, B_m)$-separator bags of $D'$ such that
	$1\le a_1<a_2<a_3<a_4 \le m$, for each $i\in \{1, \ldots, 3\}$ there are no $(B_{a_i}, B_{a_{i+1}})$-separator bags.
	If there is a bag $B_z$ where $a_2<z<a_3$, then we remove all unmarked vertices in $B_z$.
\end{RRULE}
\begin{LEM}\label{lem:bypassing3}
Reduction Rule~\ref{rrule:bypassing3} is safe.
\end{LEM}
\begin{proof}
	Suppose there is a bag $B_z$ where $a_2<z<a_3$, and let $v$ be an unmarked vertex of $B_z$.
	We claim that there is no DH obstruction containing $v$, which shows that $v$ can be safely removed.
	Suppose for contradiction that there is a DH obstruction $F$ containing $v$.
	Since the set of all unmarked vertices in $D'$ induces a distance-hereditary graph, 
	$F$ must contain at least one vertex from $S$. Notice that $\eta(B_{a_2})\cup \eta(B_{a_3})$ is a separator between $v$ and $S$ and since 
	any DH obstruction is 2-connected, $F$ contains at least two vertices of $\eta(B_{a_2})\cup \eta(B_{a_3})$. That each $\eta(B_{a_i})$ is a twin set in $G$ while $F$ contains no twins imply that 
	$F$ contains exactly vertex from $\eta(B_{a_2})\cup \eta(B_{a_3})$ 
	respectively. A similar argument shows $F$ contains exactly one vertex from $\eta(B_{a_i})$ for each $i\in \{1,2,3,4\}$.
	Furthermore, as every vertex in $S$ cannot be adjacent to both $B_{a_1}$ and $B_{a_4}$, 
	we have $\abs{F}\ge 7$, and thus $F$ is an induced cycle.

	For each $i\in \{1,2,3,4\}$, let $w_i$ be the vertex in $V(F)\cap \eta(B_i)$.

	If $v$ is contained in a complete bag, then $F$ contains an induced cycle of length $3$ together with vertices $w_2$ and $w_3$, contradiction. 
	We may assume $v$ is contained in a star bag. So, $v$ is adjacent to one of $w_2$ and $w_3$. 
	Since $F$ has no leaves, there is a neighbor $v'$ of $v$ in $F$ that is not in $\{w_2, w_3\}$. 
	Since $v'$ is also contained in a star bag, $v'$ is adjacent to one of $w_2$ and $w_3$ that is not adjacent to $v$. 
	However, this implies that $F$ contains an induced cycle of length $4$, contradiction.
	Therefore, $v$ cannot be contained any DH obstruction, and we can safely remove it.
	\end{proof}

	We can observe that after applying Reduction Rule~\ref{rrule:bypassing3} exhaustively, 
	if $B_{a_1}, B_{a_2}, \ldots, B_{a_t}$ be the sequence of $(B_1, B_m)$-separator bags where $t\ge 4$ and $a_1<a_2< \cdots <a_t$, 
	then for each $i\in \{2, \ldots, t-2\}$, $a_{i+1}=a_i+1$, that is, there are no bags between $B_{a_i}$ and $B_{a_{i+1}}$.
	We describe how we reduce the number of $(B_1, B_m)$-separator bags.

\begin{figure}
\includegraphics[scale=0.35]{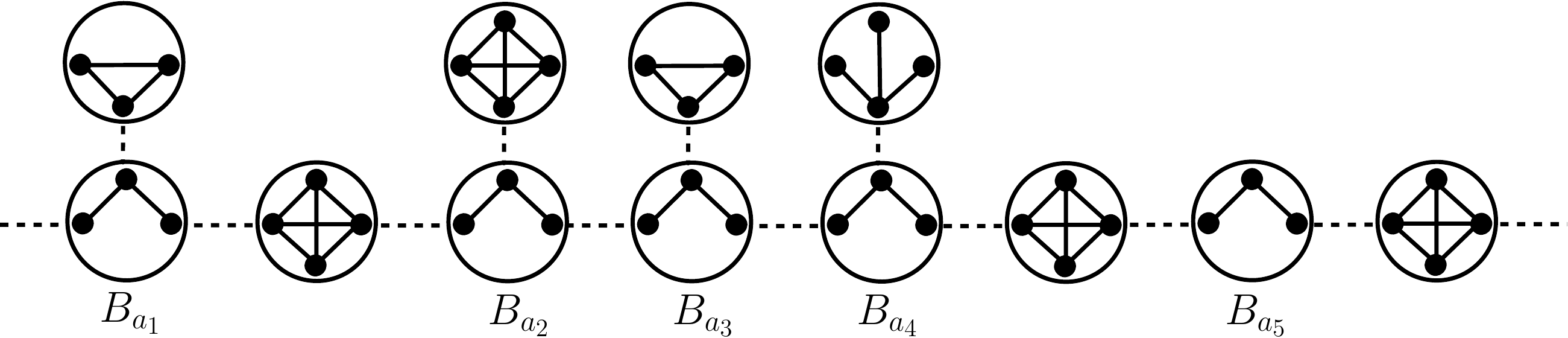} \vskip  0.3cm
\includegraphics[scale=0.35]{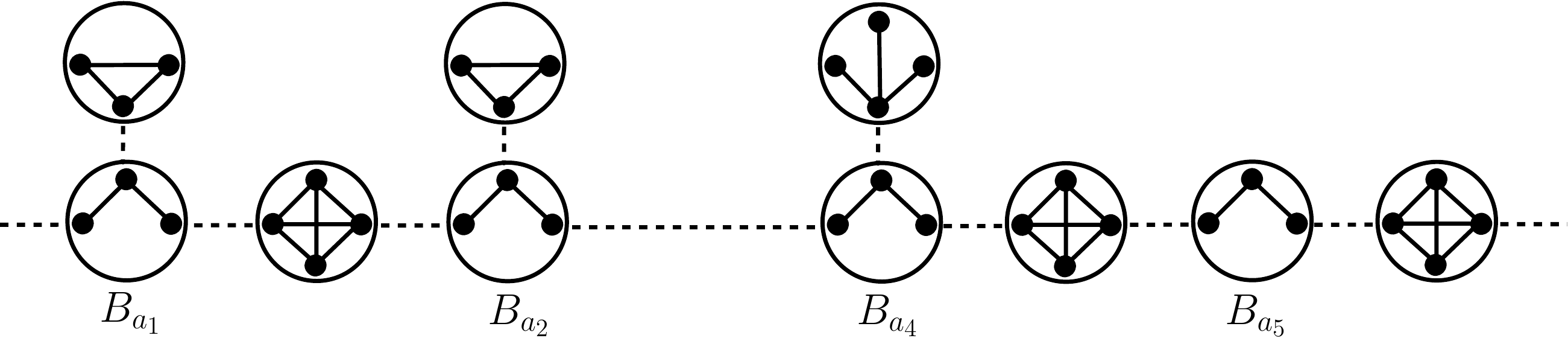}
\caption{An illustration of Reduction Rule~\ref{rrule:bypassing2}. }
\label{fig:bypassing2}
\end{figure}

\begin{RRULE}\label{rrule:bypassing2}
	Let $B_{a_1}, B_{a_2}, B_{a_3}, B_{a_4}, B_{a_5}$  be $(B_1, B_m)$-separator bags of $D'$ such that
	$1\le a_1<a_2<\cdots <a_5 \le m$, and
	for each $i\in \{1,\ldots,4\}$, there are no $(B_{a_i}, B_{a_{i+1}})$-separator bags.
	We first remove $B_{a_3}$ and also remove the leaf bag adjacent to $B_{a_3}$ if one exists, and
	link marked vertices of $B_{a_2}$ and $B_{a_4}$ in $D'$, which were adjacent to $B_{a_3}$, by a marked edge. 
	If $\abs{\eta(B_2)}>\abs{\eta(B_3)}=1$, then we remove the leaf bag adjacent to $B_{a_2}$, and if $\abs{\eta(B_2)}>\abs{\eta(B_3)}\ge 2$, 
	then we remove vertices in $\eta(B_{a_2})$ except $\abs{\eta(B_{a_3})}$ many vertices. 
\end{RRULE}

\begin{LEM}
\label{lem:bypassing2}
Reduction Rule~\ref{rrule:bypassing2} is safe. Furthermore, in the resulting graph, $S$ is again a good DH-modulator.
\end{LEM}
\begin{proof}
	Let $D''$ be the resulting canonical split decomposition, and let $G'$ be the resulting graph. 
	For each $i\in \{1, \ldots, 5\}$, let $U_i:=\eta(B_i)$.
	In the decomposition $D''$, we define $\widetilde{U}$ as 
	\begin{itemize}
	\item $\widetilde{U}:=\{v\}$ if $v$ is the center of $B_{a_2}$ and there is no leaf bag adjacent to $B_{a_2}$, 
	\item $\widetilde{U}:=X$ if there is a leaf bag $B$ adjacent to $B_{a_2}$, and $X$ is the set of unmarked vertices in $B$.
	\end{itemize}
	Notice that $G-(U_2\cup U_3)=G'-\widetilde{U}$.
	We prove that $(G,k)$ is a \YES-instance if and only if $(G', k)$ is a \YES-instance.
	
	First suppose that $G$ has a vertex set $T$ such that $\abs{T}\le k$ and $G-T$ is distance-hereditary.
	Assume that $T$ contains no vertex in $U_2\cup U_3$. We claim  that $G'-T$ is distance-hereditary.
	Suppose for contradiction that $G'-T$ contains a DH obstruction $F$.
	If $F$ does not contain a vertex in $\widetilde{U}$, then $F$ is an induced subgraph of $G-T$, as $G-(U_2\cup U_3)=G'-\widetilde{U}$. 
	Thus, $F$ contains a vertex in $\widetilde{U}$, and since $\widetilde{U}$ is a set of pairwise twins in $G'$, we have $\abs{V(F)\cap \widetilde{U}}=1$.
	Let $v$ be the vertex in $V(F)\cap \widetilde{U}$.
	We observe that 
	\begin{itemize}
	\item $F$ contain at least one vertex from each of $U_1, U_4, U_5$,
	\item the vertices in $V(F)\cap \tilde{U}$ and $V(F)\cap U_4$ have no neighbors in $S$, 
	\item every vertex in $S$ is not adjacent to both $\eta(B_1)$ and $\eta(B_5)$.
	\end{itemize}
	These imply that $F$ is an induced cycle of length at least $6$. 
	Thus we can obtain an induced cycle of length at least $7$ in $G$ from $F$ by replacing $v$ with a vertex of $U_2$ and a vertex of $U_3$, contradiction.
	We conclude that $G'-T$ is distance-hereditary when $T\cap (U_2\cup U_3)=\emptyset$.
	
	Hence, we assume $T\cap (U_2\cup U_3)\neq \emptyset$. As $U_i$ is a set of pairwise twins in $G$, 
	if $T\cap U_i\neq \emptyset$, then $U_i\subseteq T$. We can observe that $(T\setminus (U_2\cup U_3))\cup \widetilde{U}$ is a solution to $(G', k)$ and 
	we have $\abs{(T\setminus (U_2\cup U_3))\cup \widetilde{U}}\le \abs{T}\le k$. It implies that $(G', k)$ is a \YES-instance.

	For the converse direction, suppose that $G'$ has a vertex set $T'$ such that $\abs{T'}\le k$ and $G'-T'$ is distance-hereditary.
	We first assume that $T'\cap \widetilde{U}=\emptyset$.
	Suppose $G-T'$ has a DH obstruction $F$. 
	We have $V(F)\cap (U_2\cup U_3)\neq \emptyset$, otherwise $F$ is an induced subgraph of $G'-T'$.
	As $U_2$ and $U_3$ are twin sets of $G$, $F$ contains at most one vertex from each of $U_2$ and $U_3$. 
	Furthermore, $F$ contains exactly one vertex from each of $U_1, \ldots, U_5$, and thus it is an induced cycle of length at least $7$.
	We can obtain an induced cycle of length at least $6$ in $G'-T'$ from $F$ by contracting the edge between $V(F)\cap U_2$ and $V(F)\cap U_3$, 
	which contradicts to the assumption that $G'-T'$ is distance-hereditary.
	We conclude that $G-T'$ is distance-hereditary in the case when $T'\cap \widetilde{U}=\emptyset$.

	Lastly, suppose $T'\cap \widetilde{U}\neq \emptyset$. As $\widetilde{U}$ is a set of pairwise twins in $G'$, we have $\widetilde{U}\subseteq T'$.
	We obtain a set $T$ from $T'$ by removing $\widetilde{U}$, and adding $U_i$ with $i\in \{2,3\}$ where $\abs{U_i}=\min\{\abs{U_2}, \abs{U_3}\}$. 
	If $\abs{U_2}=\abs{U_3}$, then we add one of them chosen arbitrarily.
	Clearly, $\abs{T}\le \abs{T'}\le k$.
	We can observe that $G-T$ has no DH obstruction containing a vertex in $U_2\cup U_3$.
	Since $G-(U_2\cup U_3)=G'-\widetilde{U}$ has no DH obstructions, 
	we conclude that $G-T$ is distance-hereditary, as required.

	This proves that Reduction Rule~\ref{rrule:bypassing2} is safe.
	We argue that $S$ is again a good DH-modulator in $G'$.
	We need to verify that for every $v\in S$, $G'-S\setminus \{v\}$ has no DH obstructions.
	Suppose for contradiction that $G'-S\setminus \{v\}$ contains a  DH obstruction $F$ for some $v\in S$.
	The obstruction $F$ should contain a vertex in $\widetilde{U}$, otherwise, $F$ is also a DH obstruction in $G[(V(G)\setminus S)\cup \{v\}]$.
	Observe that $F$ should contain at least one vertex from $U_1, U_4, U_5$, and thus 
	$F$ is an induced cycle of length at least $6$, and the neighbors of $v$ in $F$ are contained in $U_1$ and $U_4$.
	Thus we can obtain an induced cycle of length at least $7$ by replacing the vertex in $V(F)\cap \widetilde{U}$ with a vertex of $U_2$ and a vertex of $U_3$ in $G$, 
	which implies that $G[(V(G)\setminus S)\cup \{v\}]$ contains a DH obstruction. It contradicts to our assumption that $S$ is a good DH-modulator. 
	We conclude that $S$ is a good DH-modulator in the resulting graph.
\end{proof}

	In the following reduction rule, 
	we describe how to reduce a sequence of non-$(B_1, B_m)$-separator bags.
\begin{RRULE}\label{rrule:bypassing1}
	Let $a$ be a positive integer such that   $a\le m-(5k+11)$, and 
	for each $i\in \{1, \ldots, 5k+10\}$, $B_{a+i}$ is not a $(B_1, B_m)$-separator bag.
	Among bags in $\{B_{a+i}:1\le i\le 5k+10\}$, we mark up to $k+1$ bags for each types : 
	complete bags, star bags $B_j$ whose centers are adjacent to $B_{j+1}$, and 
	star bags $B_j$ whose centers are adjacent to $B_{j-1}$.
	Choose a bag $B_{a+j}$ with $j\in \{1, \ldots, 5k+10\}$ that is not marked after finishing the marking procedure, and remove vertices in $B_{a+j}$ from $G$.
\end{RRULE}

	We observe the every DH obstruction can be turned into an induced path by removing a vertex $v$.  
	Note that this vertex is not unique; for instance, we can choose any vertex in an induced cycle of length $5$. 
	We will use this observation. 
	
\begin{LEM}
\label{lem:bypassing1}
Reduction Rule~\ref{rrule:bypassing1} is safe.
\end{LEM}
\begin{proof}
	Suppose there is a bag $B_{a+j}$ with $j\in \{1, \ldots, 5k+10\}$
	that is not marked after finishing the marking procedure in Reduction Rule~\ref{rrule:bypassing1}, and let $v$ be an unmarked vertex in $B_{a+j}$. 
	We prove that $(G,k)$ is a \YES-instance if and only if $(G-v, k)$ is a \YES-instance.
	The forward direction is clear. Suppose that $G-v$ has a vertex set $T$ with $\abs{T}\le k$ such that $(G-v)-T$ is distance-hereditary and $G-T$ contains a DH obstruction $F$.
	Note that $v\in V(F)$. Let $b:=a+(5k+11)$.

	We partition the vertex set $V(G)$ into four parts. 
	Let $U$ be the set of all unmarked vertices contained in $B_{a+1}, B_{a+2}, \ldots, B_{b-1}$.
	Since $D'$ consists of only blue bags, 
	there are exactly two sets $N_1$ and $N_2$ in $\{N_G(x)\cap U:x\in V(G)\setminus U\}$ 
	that correspond to marked edges $e(B_a, B_{a+1})$ and $e(B_{b-1}, B_b)$, respectively.
	We define
	\begin{itemize}
	\item $A_i:=\{x\in V(G)\setminus U:N_G(x)\cap U=N_i\}$ for each $i\in \{1,2\}$.
	\end{itemize}
	It is clear that $A_1$ and $A_2$ are disjoint.
	Let $W:=V(G)\setminus (U\cup A_1\cup A_2)$.
	
	We first show some necessary lemmas.

\begin{CLAIM}\label{claim:longcycle}
	If $F$ is an induced cycle of length at least $5$, then $\abs{V(F)\cap U}\le 2$.
\end{CLAIM}
\begin{proofofclaim}
	Suppose $F$ is an induced cycle of length at least $5$, and $F$ contains at least three vertices $w_1, w_2, w_3$ in $U$.
	Since each $w_i$ is contained in either $N_1$ or $N_2$, 
	at least two vertices of $w_1, w_2, w_3$ are contained in the same set of $N_1$ and $N_2$.
	By relabeling if necessary, we assume that $w_1$ and $w_2$ are contained in the same set, say $N_1$. 
	Let $B_x$ and $B_y$ be the two bags containing $w_1$ and $w_2$, respectively, 
	and without loss of generality, we may assume $x\le y$.
	Then every neighbor of $w_1$ in $G$ is adjacent to $w_2$, and therefore, $F$ contains a cycle of length $4$ as a subgraph.
	This contradicts to the fact that $F$ is an induced cycle of length at least $5$. We conclude that $\abs{V(F)\cap U}\le 2$.
\end{proofofclaim}
\begin{CLAIM}\label{claim:smallobs}
	If $F$ is a DH obstruction that is not an induced cycle, then $\abs{V(F)\cap U}\le 4$.
\end{CLAIM}
\begin{proofofclaim}
	If $\abs{V(F)\cap U}\ge 5$, then $\abs{V(F)\cap S}\le 1$, contradicting to the assumption that $S$ is a good DH-modulator.
\end{proofofclaim}

\begin{CLAIM}\label{claim:creatingobs1}
	If there are two vertices $v_1\in A_1$ and $v_2\in A_2$ such that $v_1v_2\notin E(G)$, then 
	for every $i\in \{a+1, a+2, \ldots, b-5\}$,
	bags $B_i$, $B_{i+1}$, $B_{i+2}$, $B_{i+3}$, $B_{i+4}$ contain three vertices $w_1, w_2, w_3$ where 
	$G[\{v_1, v_2, w_1, w_2, w_3\}]$ is isomorphic to a DH obstruction.	
\end{CLAIM}
\begin{proofofclaim}
	Suppose there exist $v_1\in A_1$ and $v_2\in A_2$ such that $v_1v_2\notin E(G)$. 
	Let $i\in \{a+1, a+2, \ldots, b-5\}$.
	We prove three special cases.

\medskip
\noindent\textbf{Case 1.} (There exist $x,y, z\in \{i, i+1, \ldots, i+4\}$ with $x<y$ such that $B_x$ is a star bag whose center is adjacent to $B_{x+1}$, $B_y$ is a star bag whose center is adjacent to $B_{y-1}$, 
and $B_z$ is a complete bag.)

	Let $w_1, w_2, w_3$ be unmarked vertices of $B_x, B_y, B_z$, respectively.
	Since $v_1w_2w_1v_2$ is an induced path and $w_3$ is adjacent to both $v_1$ and $v_2$, 
	by Lemma~\ref{lem:createdhobs}, $G[\{v_1, v_2, w_1, w_2, w_3\}]$ is isomorphic to a DH obstruction.

\medskip
\noindent\textbf{Case 2.} (There exist $x,y, z\in \{i, i+1, \ldots, i+4\}$ with $x<y<z$ such that $B_x$ and $B_z$ are complete bags, and $B_y$ is a star bag.)

	Let $w_1, w_2, w_3$ be unmarked vertices of $B_x, B_y, B_z$, respectively.
	Assume $B_y$ is a star bag whose center is adjacent to $B_{y-1}$.
	Since $w_2v_1w_3v_2$ is an induced path and $w_1$ is adjacent to both $w_2$ and $v_2$, 
	by Lemma~\ref{lem:createdhobs}, $G[\{v_1, v_2, w_1, w_2, w_3\}]$ is isomorphic to a DH obstruction.
	If $B_y$ is a star bag whose center is adjacent to $B_{y+1}$, 
	then $w_2v_2w_1v_1$ is an induced path and $w_3$ is adjacent to both $w_2$ and $v_1$,
	 $G[\{v_1, v_2, w_1, w_2, w_3\}]$ is isomorphic to a DH obstruction.
	 
\medskip
\noindent\textbf{Case 3.} (There exist $x_1, x_2, x_3, x_4\in \{i, i+1, \ldots, i+4\}$ with $x_1<x_2<x_3<x_4$ such that 
	$B_{x_1}$ and $B_{x_3}$ are star bags whose centers are adjacent to the next bags,
	and $B_{x_2}$ and $B_{x_4}$ are star bags whose centers are adjacent to the previous bags.)

	Let $w_1, w_2, w_3, w_4$ be unmarked vertices of $B_{x_1}, B_{x_2}, B_{x_3}, B_{x_4}$, respectively.
	It is not hard to verify that $G[\{v_1, v_2, w_1, w_2, w_3, w_4\}]$ is isomorphic to the domino.

\medskip
	 Now we prove in general.
	 Since $D$ is a canonical split decomposition, 
	 one of $B_i$ and $B_{i+1}$ is either a complete bag or a star bag whose center is adjacent to the next bag.
	 Let $B_{i'}$ be such a bag.
	 Assume $B_{i'}$ is a complete bag. 
	 Then $B_{i'+1}$ is a star bag.
	 If there is a complete bag among $B_{i'+2}, B_{i'+3}$, 
	 then the claim holds by \textbf{Case 2}.
	 Otherwise, all of $B_{i'+1}, B_{i'+2}, B_{i'+3}$ are star bags, 
	 and the claim holds by \textbf{Case 1}.

	Assume $B_{i'}$ is a star bag whose center is adjacent to $B_{i'+1}$.
	If all of $B_{i'+1}, B_{i'+2}, B_{i'+3}$ are star bags, then 
	the claim holds by \textbf{Case 3}.
	We may assume there exists a complete bag in $B_{i'+1}, B_{i'+2}, B_{i'+3}$.
	If two other bags are star bags whose centers are adjacent to the next bags, 
	then the claim holds by \textbf{Case 2}.
	Otherwise, there is a star bag whose center is adjacent to the previous bag, 
	and the claim holds by \textbf{Case 1}.
\end{proofofclaim}

\begin{CLAIM}\label{claim:creatingobs2}
	If there are three vertices $v_1\in A_1, v_2\in A_2, v_3\in W$ such that $v_1v_2, v_2v_3, v_3v_1\in E(G)$, then 
	for every $i\in \{a+1, a+2, \ldots, b-4\}$,
	bags $B_i$, $B_{i+1}$, $B_{i+2}$, $B_{i+3}$ contain two vertices $w_1, w_2$ where 
	$G[\{v_1, v_2, v_3, w_1, w_2\}]$ is isomorphic to a DH obstruction.	
\end{CLAIM}
\begin{proofofclaim}
	Suppose there are three vertices $v_1\in A_1, v_2\in A_2, v_3\in W$ such that $v_1v_2, v_2v_3, v_3v_1\in E(G)$.
	Let $i\in \{a+1, a+2, \ldots, b-4\}$.
	We prove three special cases.
	We observe that there is $x\in \{i, i+1, i+2, i+3\}$ such that either
	\begin{enumerate}[(1)]
	\item $B_x$ is a complete bag and $B_{x+1}$ is a star bag whose center is adjacent to $B_{x}$, 
	\item $B_x$ is a star bag whose center is adjacent to $B_{x+1}$ and $B_{x+1}$ is a complete bag, or
	\item $B_x$ and $B_{x+1}$ are star bags and their centers are adjacent.
	\end{enumerate}
	Since $D$ is a canonical split decomposition, 
	 one of $B_i$ and $B_{i+1}$ is either a complete bag or a star bag whose center is adjacent to the next bag.
	 Let $B_{i'}$ be such a bag.
	 Assume $B_{i'}$ is a complete bag. 
	 Then $B_{i'+1}$ is a star bag.
	 If its center is adjacent to $B_{i'}$, then the statement (1) holds.
	 We may assume the center of $B_{i'+1}$ is adjacent to $B_{i'+2}$.
	 If $B_{i'+2}$ is a complete bag, then the statement (2) holds, 
	 and if $B_{i'+2}$ is a star bag whose center is adjacent to $B_{i'+1}$, then the statement (3) holds. 
	Assume $B_{i'}$ is a star bag whose center is adjacent to $B_{i'+1}$.
	If $B_{i'+1}$ is a complete bag, then the statement (2) holds, and 
	if $B_{i'+1}$ is a star bag whose center is adjacent to $B_{i'}$, then the statement (3) holds.
	
	Let $w_1\in B_x$, $w_2\in B_{x+1}$ be unmarked vertices for such bags $B_x$ and $B_{x+1}$. 	
	One can observe that in any case, $G[\{v_1, v_2, v_3, w_1, w_2\}]$ is isomorphic to a DH obstruction.	
\end{proofofclaim}

\begin{CLAIM}\label{claim:creatingobs3}
	If there are four vertices $v_1\in A_1, v_2\in A_2, v_3, v_4\in W$ such that $v_1v_2, v_2v_3, v_3v_4, v_4v_1\in E(G)$ and $v_1v_3, v_2v_4\notin E(G)$, then 
	for every $i\in \{a+1, a+2, \ldots, b-3\}$,
	bags $B_i$, $B_{i+1}$, $B_{i+2}$ contain two vertices $w_1, w_2$ where 
	$G[\{v_1, v_2, v_3, v_4, w_1, w_2\}]$ contains an induced subgraph isomorphic to a DH obstruction.	
\end{CLAIM}
\begin{proofofclaim}
	Suppose there are four vertices $v_1\in A_1, v_2\in A_2, v_3, v_4\in W$ such that $v_1v_2, v_2v_3, v_3v_4, v_4v_1\in E(G)$ and $v_1v_3, v_2v_4\notin E(G)$.
	Let $i\in \{a+1, a+2, \ldots, b-3\}$.
	If $B_i, B_{i+1}, B_{i+2}$ contains a complete bag, then an unmarked vertex $w$ in the bag satisfies that 
	$G[\{v_1, v_2, v_3, v_4, w\}]$ is isomorphic to the house.
	We may assume those bags are star bags.
	Then there is $i'\in \{i, i+1\}$ such that 
	$B_{i'}$ is a star bag whose center is adjacent to $B_{i'+1}$ and $B_{i'+1}$ is a star bag whose center is adjacent to $B_{i'}$.
	 Let $w_1\in B_{i'}$, $w_2\in B_{i'+1}$ be unmarked vertices. 	
	One can observe $G[\{v_1, v_2, v_3, v_4, w_1, w_2\}]$ is isomorphic to the domino.	
\end{proofofclaim}

 	Now, we prove the result based on the previous claims.
	Suppose $\abs{V(F)\cap U}=1$.
	In this case, by the marking procedure of Reduction Rule~\ref{rrule:bypassing1}, 
	there are distinct integers $j_1, j_2, \ldots, j_{k+1}\in \{1, \ldots, 5k+10\}\setminus \{j\}$ such that 
	$B_{a+j}$, $B_{a+j_1}$, $B_{a+j_2}, \ldots, B_{a+j_{k+1}}$ have the same type (recall that $B_{a+j}$ is an unmarked bag in the application of Reduction~\ref{rrule:bypassing1} abd $v\in B_{a+j}$).
	Therefore, there is a vertex $v'$ not contained in $T$, where $v$ and $v'$ have the same neighborhood on $V(F)\setminus \{v\}$, 
	which implies that $G[(V(F)\setminus \{v\})\cup \{v'\}]$ is a DH obstruction of $(G-v)-T$.
	This is contradiction. Thus, we may assume that $\abs{V(F)\cap U}\ge 2$. 

	We choose $q\in V(F)$ such that $F-q$ is an induced path, and
	\begin{itemize}
	\item $q=v$ if $F$ is an induced cycle, 
	\item $q$ is the closest vertex to $v$ in the underlying cycle if $F$ is either the house, the gem, or the domino.
	\end{itemize}
	Let $P:=F-q$, and let $w$ and $z$ be the end vertices of $P$, and
	let $w'$ and $z'$ be the neighbors of $w$ and $z$ in $P$, respectively. 

	We divide into following four cases depending on the places of $w$ and $z$:
	\begin{enumerate}[(1)]
	\item Both $w$ and $z$ are contained in $U$.
	\item One of $w$ and $z$ is in $A_1\cup A_2$ and the other is contained in $U$.
	\item Both $w$ and $z$ are contained in $A_1\cup A_2$.
	\item One of $w$ and $z$ is contained in $W$.
	\end{enumerate}
	We aim to show that these cases are not possible, because $(G-v)-T$ is distance-hereditary.

	First observe that $w$ and $z$ are not contained in $U$ together.
	Suppose $w$ and $z$ are contained in $U$.  Let $B_x$ and $B_y$ be the bags containing $w$ and $z$, respectively.
	We can assume $x\le y$, otherwise the proof is symmetric.
	Since they are not adjacent, either $B_x$ is a star whose center is adjacent to $B_{x-1}$ or $B_y$ is a star whose center is adjacent to $B_{y+1}$.
	By symmetry, assume that $B_y$ is a star bag whose center is adjacent to $B_{y+1}$. 
	Since $q$ is adjacent to both $w$ and $z$, 
	$B_x$ should be either a complete bag, or a star bag whose center is adjacent to $B_{x+1}$.
	Then $z'$ should be also adjacent to $w$, contradicting to the fact that $F-q$ is an induced path.

	Secondly, assume that one of $w$ and $z$ is in $A_1\cup A_2$ and the other is contained in $U$.
	We assume that $w\in A_1$ and $z\in U$.
	For the other cases ($w\in A_2$ and $z\in U$) and ($z\in A_1$ and $w\in U$) and ($z\in A_2$ and $w\in U$), 
	we can prove in the similar way.
	Since $wz\notin E(G)$, $z$ is contained in a star bag $B_x$ whose center is adjacent to $B_{x+1}$ for some $x\in \{a+1, a+2, \ldots, b-1\}$.
	Furthermore, if $z'\in U$, then $z'$ should be adjacent to $w$, a contradiction.
	Thus, we have $z'\in A_2$. Note that $w, z'\notin T$ and $wz'\notin E(G)$.
	Since $b-a-1\ge 5(k+2)$, by Claim~\ref{claim:creatingobs1}, 
	there exists $i\in \{a+1, a+2, \ldots, b-5\}$ such that
	bags $B_i$, $B_{i+1}$, $B_{i+2}$, $B_{i+3}$, $B_{i+4}$ contain three vertices $w_1, w_2, w_3$ where 
	$w_1, w_2, w_3\notin T\cup \{v\}$ and $G[\{w, z', w_1, w_2, w_3\}]$ is isomorphic to a DH obstruction.
	This contradicts to the assumption that $(G-v)-T$ is distance-hereditary.	

	Thirdly, assume that both $w$ and $z$ are contained in $A_1\cup A_2$.
	If $w$ and $z$ are contained in distinct sets of $A_1$ and $A_2$, 
	then by the same argument in the previous paragraph, 
	there exists $i\in \{a+1, a+2, \ldots, b-5\}$ such that 
	bags $B_i$, $B_{i+1}$, $B_{i+2}$, $B_{i+3}$, $B_{i+4}$ contain three vertices $w_1, w_2, w_3$ where 
	$w_1, w_2, w_3\notin T\cup \{v\}$ and $G[\{w, z, w_1, w_2, w_3\}]$ is isomorphic to a DH obstruction.
	Thus, we may assume both $w$ and $z$ are contained in $A_i$ for some $i\in \{1,2\}$.
	Without loss of generality, we assume $w,z\in A_1$.
	Suppose there are at least $k+5$ bags $B_x$ in $\{B_{a+1}, B_{a+2}, \ldots, B_{b-1}\}$ 
	that are either complete bags or star bags whose centers are adjacent to $B_{x-1}$. 
	Note that by Claims~\ref{claim:longcycle} and \ref{claim:smallobs}, 
	$\abs{V(F)\cap U}\le 4$.
	Thus, there is a bag among them having no vertex in $T\cup V(F)$, and
	an unmarked vertex in the bag will form a DH obstruction with $P$ in $(G-v)-T$.
	Therefore, there are at most $k+4$ such bags.
	This implies that there are at most $2k+9$ bags, as there are no two consecutive bags in the canonical split decomposition.
	It contradicts to our assumption that $b-a-1\ge 5k+10$.
	
	Lastly, we assume that one of $w$ and $z$ is contained in $W$.
	If  $q=v$, then $w$ or $z$ cannot be in $W$.
	Especially, $F$ is a DH obstruction that is not an induced cycle.
	Let $q'$ be the vertex of degree $3$ other than $q$ if $F$ is the domino.
	Let us assume $w\in W$. The proof is symmetric when $z\in W$.

	If $q\in W$, then $\{q,w,w',z\}\subseteq A_1\cup A_2\cup W$.
	Since $\abs{V(F)\cap U}\ge 2$, $F$ should be isomorphic to the domino and $V(F)\cap U=\{q', z'\}$.
	But since $q$ is adjacent to $q'$, this is contradiction. 
	Therefore, we have $q\in A_1\cup A_2$.	
	We remark that if $V(F)\cap A_i= \emptyset$ for some $i\in \{1,2\}$, then $(V(F)\cap U, V(F)\cap (W\cup A_{3-i}))$ is a split of $F$, contradicting the fact that $F$ has no splits. Hence, 
	it holds that $V(F)\cap A_i\neq \emptyset$ for  $i\in \{1,2\}$.
	We divide cases depending on whether $F$ is isomorphic to the domino or not.

	\medskip
	\noindent\textbf{Case 1.} $F$ is not isomorphic to the domino.
	
	Note that $\{q,w,w'\}\subseteq A_1\cup A_2\cup W$.
	Since $\abs{V(F)}=5$ and $\abs{V(F)\cap U}\ge 2$, we have $V(F)\cap U=\{z,z'\}$. 
	Furthermore, $q$ and $w'$ are contained in distinct sets of $A_1$ and $A_2$.
	By Claim~\ref{claim:creatingobs2}, 
	for every $i\in \{a+1, a+2, \ldots, b-4\}$,
	$B_{i}, B_{i+1}, \ldots, B_{i+3}$ contain two vertices $w_1$ and $w_2$ where
	$G[\{q,w,w',w_1,w_2\}]$ is isomorphic to a DH obstruction.
	As $b-a-1\ge 4(k+2)$, there is a pair of such vertices $w_1$ and $w_2$ not contained in $T\cup \{v\}$.
	So, $(G-v)-T$ contains a DH obstruction, which is contradiction.

	\medskip
	\noindent\textbf{Case 2.} $F$ is isomorphic to the domino.

	Recall that $q$ is chosen as the vertex of degree $3$ in $F$ closer to $v$.
	Since $q\notin U$, $v$ is a vertex of degree $2$.
	Furthermore, since $w\in W$, we have $z=v$.
	
	Note that $w'\in W\cup A_1\cup A_2$.
	Suppose $q$ and $w'$ are contained in distinct sets of $A_1$ and $A_2$.
	Since $b-a-1\ge 5(k+2)$, by Claim~\ref{claim:creatingobs1}, 
	there exists $i\in \{a+1, a+2, \ldots, b-5\}$ such that
	bags $B_i$, $B_{i+1}$, $B_{i+2}$, $B_{i+3}$, $B_{i+4}$ contain three vertices $w_1, w_2, w_3$ where 
	$w_1, w_2, w_3\notin T\cup \{v\}$ and $G[\{q, w', w_1, w_2, w_3\}]$ is isomorphic to a DH obstruction.
	If $q$ and $w'$ are contained in the same set of $A_1$ or $A_2$, 
	then $z$ should be adjacent to $w'$, contradiction.
	Thus, we can conclude that $w'\in W$.
	Therefore, $q'\notin U$, and $\{z,z'\}\subseteq U$.
	Since $q'$ is not adjacent to $z$, $q$ and $q'$ are contained in distinct sets of $A_1$ and $A_2$.
	
	Now, by Claim~\ref{claim:creatingobs3}, 
	for every $i\in \{a+1, a+2, \ldots, b-3\}$,
	bags $B_i$, $B_{i+1}$, $B_{i+2}$ contain two vertices $w_1, w_2$ where 
	$G[\{w, w', q, q', w_1, w_2\}]$ contains an induced subgraph isomorphic to a DH obstruction.	
	As $b-a-1\ge 3(k+2)$, there is a pair of such vertices $w_1$ and $w_2$ not contained in $T\cup \{v\}$.
	So, $(G-v)-T$ contains a DH obstruction, which is contradiction.

\medskip

We proved that $\abs{V(F)\cap U}\ge 2$ cannot hold when $(G-v, k)$ is a \YES-instance.
We conclude that $(G,k)$ is a \YES-instance if and only if $(G-v, k)$ is a \YES-instance.
\end{proof}

\subsection{The size of a non-trivial component}

\begin{LEM}\label{lem:boundpath}
Let $(G,k)$ be an instance reduced under Reduction Rules~\ref{rrule:removeleaf1}, \ref{rrule:removeleaf2}, \ref{rrule:fliptwins}, \ref{rrule:bypassing3}, \ref{rrule:bypassing2}, and \ref{rrule:bypassing1} and let $D$ be the canonical split decomposition of a non-trivial connected component of $G-S$. If $D'$ be a connected component of $D-\bigcup_{B\in \mathcal{R}\cup \mathcal{Q}}V(B)$,
then $D'$ contains at most $20k+52$ bags.
\end{LEM}
\begin{proof}
	Let $B_1, \cdots ,B_m$ be a sequence of bags in $D'$. By Reduction Rule~\ref{rrule:bypassing2}, $D'$ contains at most four $(B_1, B_m)$-separator bags $B_j$, 
	and by Reduction Rule~\ref{rrule:bypassing3}, if there are four $(B_1, B_m)$-separators bags, 
	then there are no bags between the second and third $(B_1, B_m)$-separator bags. 	
	Also, there is at most leaf bag adjacent to a $(B_1, B_m)$-separator bag, 
	and there are no leaf bags adjacent to a non-$(B_1, B_m)$-separator bag.
	By Reduction Rule~\ref{rrule:bypassing1}, there are at most $5k+11$ consecutive  bags $B_j$ that are not $(B_1, B_m)$-separator bags.
	In total, $D'$ contains at most $4(5k+11)+4\cdot 2=20k+52$ bags. 
\end{proof}

We summarized the main result of this section.

\begin{LEM}\label{lem:boundcomp}
Let $(G,k)$ be an instance reduced under Reduction Rules~\ref{rrule:removeleaf1}, \ref{rrule:removeleaf2}, \ref{rrule:fliptwins}, \ref{rrule:bypassing3}, \ref{rrule:bypassing2}, and \ref{rrule:bypassing1} and let $D$ be the canonical split decomposition of a non-trivial connected component of $G-S$. Then $D$ has at most $3\abs{S}(20k+54)$ bags.
\end{LEM}
\begin{proof}
	Lemma~\ref{lem:longbluepath} states that $\abs{\mathcal{R}\cup \mathcal{Q}}\leq 6\abs{S}$ and 
	there are at most $3\abs{S}$ connected components of $D-\bigcup_{B\in \mathcal{R}\cup \mathcal{Q}}V(B)$.
	By Lemma~\ref{lem:boundpath}, 
	each connected component of  $D-\bigcup_{B\in \mathcal{R}\cup \mathcal{Q}}V(B)$ contains at most $20k+52$ bags.
	Therefore, $D$ contains at most $6\abs{S}+ 3\abs{S}(20k+52)=3\abs{S}(20k+54)$ bags.
\end{proof}

\section{Polynomial kernel for \dhd}\label{sec:total}

We present a proof of Theorem~\ref{thm:main1} in this section. 
\begin{proof}[Proof of Theorem~\ref{thm:main1}]
We first prove that given an instance $(G,k)$ and a good \dhm\ $S$, one can output an equivalent instance of size $O(k^5\abs{S}^5)$.
We first apply Reduction Rule~\ref{rrule:boundingcc} to $(G,k)$ with $S$. 
After that, $G-S$ has $O(k^2\abs{S})$ non-trivialial connected components or we can correctly report that $(G,k)$ is 
a \NO-instance by Lemma~\ref{prop:ccnumber}. Notice that $S$ remains a good \dhm\ after the application of 
Reduction Rule~\ref{rrule:boundingcc} by Lemma~\ref{lem:boundingcc}.

We apply Reduction Rules~\ref{rrule:removeleaf1},~\ref{rrule:removeleaf2},~\ref{rrule:fliptwins},~\ref{rrule:bypassing3},~\ref{rrule:bypassing2} and \ref{rrule:bypassing1} exhaustively. 
The first three rules can be applied exhaustively in polynomial time by Lemma~\ref{lem:polytime}. Polynomial-time applicability of 
Reduction Rules~\ref{rrule:bypassing2}, and \ref{rrule:bypassing1} is straightforward. We remark that $S$ remains a 
good \dhm\ after each application of Reduction Rules~\ref{rrule:removeleaf1},~\ref{rrule:removeleaf2} and~\ref{rrule:fliptwins} by Lemma~\ref{lem:polytime}. Also Reduction Rules~\ref{rrule:bypassing3} and~\ref{rrule:bypassing1} preserves $S$ as a good \dhm\ since these rules delete vertices not contained in $S$. The application of Reduction Rule~\ref{rrule:bypassing2} preserves the goodness of $S$ by Lemma~\ref{lem:bypassing2}. Then 
the canonical split decomposition $D$ of each non-trivial connected component of $G-S$ has at most  $3\abs{S}(20k+54)$ bags by Lemma~\ref{lem:boundcomp}. 

Then apply Reduction Rule~\ref{rrule:exttwinreduction} 
exhaustively in polynomial time. This bounds the size of a twin set in $G-S$ by $O(k^2\abs{S}^3)$ by Lemma~\ref{lem:twinsize}. 
We note that the unmarked vertices of a bag form at most two twin sets. Therefore, the number of unmarked vertices in a bag is 
bounded by $O(k^2\abs{S}^3)$ by Lemma~\ref{lem:twinsize}. Especially, the same bounds apply to the number of trivial components 
in $G-S$ since they form an independent set in $G-S$. 

Let $(G',k')$ be the resulting instance. Combining the previous bounds altogether, we conclude that $V(G')=O(k^5\cdot \abs{S}^5)$. 

\smallskip

We may assume that the input instance $(G,k)$ satisfies $n\leq 2^{ck}$ for some constant $c$. Recall that there is an algorithm  for \dhd\ running in time $2^{ck}\cdot n^{\mathcal{O}(1)}$ by Eiben, Ganian, and Kwon~\cite{EibenGK2016}. If $n> 2^{ck}$, then the algorithm of~\cite{EibenGK2016} solves 
the instance $(G,k)$ correctly in  polynomial time, in which case we can output a trivial equivalent instance. 
By Theorem~\ref{thm:goodmodulator}, we can obtain a good \dhm\ $S$ of size $O(k^5\log n)=O(k^6)$ in polynomial time or correctly report $(G,k)$ as a \NO-instance. 

The previous argument yields that in polynomial time, an equivalent instance $(G',k')$ of size $O(k^{35})$ can be constructed. Now, applying Theorem~\ref{thm:goodmodulator} again\footnote{That applying kernelization twice can yield an improved bound was adequately observed in~\cite{AgrawalLMSZ17}.} to $(G',k')$, we can either correctly conclude that $(G',k')$, and thus $(G,k)$,  is a \NO-instance or output a good \dhm\ $S'$ of size $O(k^5\log k)$. Now we obtain a kernel of size $O(k^{30}\cdot \log^5 k)$. 
\end{proof}

\section{Concluding remarks}\label{sec:conclusion}

Apparently, there is much room to improve the kernel size $O(k^{30}\log^5 k)$ presented in this work. 
It is not difficult to convert our approximation algorithm in Theorem~\ref{thm:approx} to an $O({\sf opt}^2\log n)$-approximation algorithm. 
It is an intriguing question to obtain an approximation algorithm with better performance ratio. This will immediately improve our kernelization bound of this paper. 

Given a \dhm\ $S$, Proposition~\ref{prop:sunflower} states that extra factor of $O(k^2)$ will be incurred per vertex in $S$ in the course of obtaining a good \dhm. In fact, a good \dhm\ of size $O(k\abs{S})$ can be constructed in polynomial time using Mader's $\mathcal{S}$-path theorem. Given  a collection $\mathcal{S}$ of disjoint vertex sets in $G$, an \emph{$\mathcal{S}$-path} is a path whose end vertices belong to distinct sets in $\mathcal{S}$. Mader's $\mathcal{S}$-path theorem provides a primal-dual characterization of the maximum number of pairwise vertex-disjoint $\mathcal{S}$-paths. From $\mathcal{S}$-path theorem, one can show the following alternative to Proposition~\ref{prop:sunflower}.
 
\begin{PROP}
Let $G$ be a graph without any small DH obstruction, $v$ be a vertex of $G$ such that $G-v$ is distance-hereditary and $k$ be a positive integer. In polynomial time, one can decide whether there is $k+1$-sunflower at $v$ or find a set $X\subseteq V(G)\setminus \{v\}$ of size at most $2k$ such that $G-X$ contains no induced cycle of length at least 5 traversing $v$. 
\end{PROP}
This proposition holds because the property of distance at most two in $G-v$ between neighbors of $v$ gives an equivalent relation, and thus hitting induced cycles of length at least $5$ traversing $v$ can be translated to 
hitting all paths linking two distinct equivalent classes, and we can use Mader's $\mathcal{S}$-path Theorem.

The caveat here is that the polynomial time algorithm in this proposition calls as a subroutine an algorithm which can efficiently compute both a primal and a dual optimal solutions of Mader's characterization. Lovasz~\cite{Lovasz80} showed that $\mathcal{S}$-path packing problem and the min-max duality is a special case of linear matroid parity problem and the corresponding duality. Matroid parity problem is NP-hard even with a compact representation~\cite{Lovasz80}, but for linear matroid it can be efficiently solved. We can use such an algorithm $\mathcal{A}$, for example~\cite{Lovasz80} or~\cite{CheungLL14} for a more recent treatement. Moreover, the optimal dual solution satisfying Mader's $\mathcal{S}$-path theorem can be efficiently computed via computing an optimal dual solution for the corresponding linear matroid parity problem, see~\cite{Orlin08}.

In this paper, we give a simpler-to-describe algorithm using an approximation algorithm for \textsc{Vertex Multicut} instead of relying on a reduction to linear matroid parity problem. An interesting question is, can we efficiently find the primal and dual optimal solution satisfying Mader's $\mathcal{S}$-path theorem without using a reduction to linear matroid parity problem? We are not aware of any literature claiming such a result.
One might also ask for a kernelization lower bound for \dhd\ parameterized by the size of a \dhm\ or by $k$.

\section*{Acknowledgement}
We thank Saket Saurabh for pointing to the idea of~\cite{AgrawalLMSZ17}  to apply the kernelization twice, which leads to an improved bound in Theorem~\ref{thm:main1}.

\end{document}